\documentclass[11pt]{article}

\usepackage{amsfonts}
\usepackage{amssymb}
\usepackage{amstext}
\usepackage{amsmath}
\usepackage{enumerate}
\usepackage{algorithm}
\usepackage{algorithmic}

\usepackage{amsthm}

\usepackage{thmtools}

\usepackage[usenames,dvipsnames,svgnames,table]{xcolor}

\usepackage[pagebackref,letterpaper=true,colorlinks=true,pdfpagemode=UseNone,citecolor=OliveGreen,linkcolor=BrickRed,urlcolor=BrickRed,pdfstartview=FitH]{hyperref}

\usepackage{framed}
\usepackage{tkz-graph}
\usetikzlibrary{calc}


\newtheorem{theorem}{Theorem}[section]
\newtheorem{fact}[theorem]{Fact}
\newtheorem{lemma}[theorem]{Lemma}

\newtheorem{definition}[theorem]{Definition}

\newtheorem{corollary}[theorem]{Corollary}

\newtheorem{proposition}[theorem]{Proposition}

\newtheorem{problem}[theorem]{Problem}
\newtheorem{assumption}[theorem]{Assumption}
\newtheorem{remark}[theorem]{Remark}
\newtheorem*{remark*}{Remark}

\newtheorem{conjecture}[theorem]{Conjecture}

\newcommand{\ch}{\mathrm{\raisebox{.07cm}{$\chi$}}}

\numberwithin{equation}{section}
\numberwithin{table}{section}

\newcommand{\R}{\ensuremath{\mathbb R}}

\newcommand{\E}[1]{{\mathbb{E}}\left[#1\right]}

\newcommand{\junk}[1]{}

\newcommand{\prob}[1]{\ensuremath{\text{{\bf Pr}$\left[#1\right]$}}}

\newcommand{\ol}{\overline}

\newcommand{\vol}{{\rm vol}}

\def\b1{{\bf 1}}
\def\eps{{\varepsilon}}

\def\R{\mathbb{R}}

\def\vol{\operatorname{vol}} 
\def\SSE{\operatorname{SSE}} 
\def\ene{\mathcal E} 
\def\prob{\mathbb P} 
\def\Reff{\text{Reff}}

\global\long\def\E{\mathbb{E}}

\global\long\def\R{\mathbb{R}}

\DeclareMathOperator{\sspan}{span}

\title{Network Design for $s$-$t$ Effective Resistance}

\author{Pak Hay Chan\footnote{University of Waterloo. Email: \href{mailto:ph5chan@uwaterloo.ca}{ph5chan@uwaterloo.ca}},~~~
Lap Chi Lau\footnote{University of Waterloo. Supported by NSERC Discovery Grant 2950-120715 and NSERC Accelerator Supplement 2950-120719. Email: \href{mailto:lapchi@uwaterloo.ca}{lapchi@uwaterloo.ca}},~~~
Aaron Schild\footnote{University of California, Berkeley. Email: \href{mailto:aschild@berkeley.edu}{aschild@berkeley.edu}},~~~
Sam Chiu-wai Wong\footnote{Microsoft Research Redmond. Email: \href{mailto:samwon@microsoft.com}{samwon@microsoft.com}},~~~
Hong Zhou\footnote{University of Waterloo. Email: \href{mailto:h76zhou@uwaterloo.ca}{h76zhou@uwaterloo.ca}}
}

\date{}

\begin{document}

\begin{titlepage}
\def\thepage{}
\thispagestyle{empty}

\maketitle

\begin{abstract}
We consider a new problem of designing a network with small $s$-$t$ effective resistance.
In this problem, we are given an undirected graph $G=(V,E)$, two designated vertices $s,t \in V$, and a budget $k$.
The goal is to choose a subgraph of $G$ with at most $k$ edges to minimize the $s$-$t$ effective resistance.
This problem is an interpolation between the shortest path problem and the minimum cost flow problem and has applications in electrical network design.

We present several algorithmic and hardness results for this problem and its variants.
On the hardness side, we show that the problem is NP-hard, and the weighted version is hard to approximate within a factor smaller than two assuming the small-set expansion conjecture. 
On the algorithmic side, we analyze a convex programming relaxation of the problem and design a constant factor approximation algorithm.
The key of the rounding algorithm is a randomized path-rounding procedure based on the optimality conditions and a flow decomposition of the fractional solution.
We also use dynamic programming to obtain a fully polynomial time approximation scheme when the input graph is a series-parallel graph, 
with better approximation ratio than the integrality gap of the convex program for these graphs.
\end{abstract}

\end{titlepage}

\thispagestyle{empty}

\newpage

\section{Introduction}

Network design problems are generally about finding a minimum cost subgraph that satisfies certain ``connectivity'' requirements.
The most well studied problem is the survivable network design problem~\cite{GGP+94,AKR95,GW95,Jai01,GGT+09}, where the requirement is to have a specified number $r_{u,v}$ of edge-disjoint paths between every pair of vertices $u,v$.
Other combinatorial requirements are also well studied in the literature, including vertex connectivity~\cite{KKL04,FL08,CCK08,CK09,Lae14,CV14} and shortest path distances~\cite{DK99,DZ16}.

Some spectral requirements are also studied, including spectral expansion~\cite{KMS+10,ALS+17}, total effective resistances~\cite{GBS08,NST18}, and mixing time~\cite{BDX04}, but in general much less is known about these problems. 
See Section~\ref{ss:results} for more discussions of previous work.

In this paper, we study a basic problem in designing networks with a spectral requirement -- the effective resistance between two vertices.

\begin{definition}[The $s$-$t$ effective resistance network design problem]
\label{d:problem}
The input is an undirected graph $G=(V,E)$,
two specified vertices $s, t \in V$, and a budget $k$.
The goal is to find a subgraph $H$ of $G$ with at most $k$ edges that minimizes $\Reff_H(s,t)$, where $\Reff_H(s,t)$ denotes the effective resistance between $s$ and $t$ in the subgraph $H$.
See Section~\ref{ss:electric} for the definition of effective resistance and Section~\ref{ss:convex} for a mathematical formulation of the problem.
\end{definition}

The $s$-$t$ effective resistance is an interpolation between $s$-$t$ shortest path distance and $s$-$t$ edge connectivity. 
To see this, let $f \in \mathbb{R}^{|E|}$ be a unit $s$-$t$ flow in $G$ and
define the $\ell_p$-energy of $f$ as $\ene_p(f) := (\sum_e |f_e|^p)^{1/p}$,
and let $\ene_p(s,t) := \min_f \{\ene_p(f)~|~f {\rm~is~a~unit~}s$-$t {\rm~flow}\}$ be the minimum $\ell_p$-energy of a unit $s$-$t$ flow that the graph $G$ can support.
Thomson's principle (see Section~\ref{ss:electric}) states that $\Reff_G(s,t) = \ene^2_2(s,t)$, so that a graph of small $s$-$t$ effective resistance can support a unit $s$-$t$ flow with small $\ell_2$-energy.
Note that the shortest path distance between $s$ and $t$ is $\ene_1(s,t)$ (as
the $\ell_1$-energy of a flow is just the average path length and is minimized by a shortest $s$-$t$ path), and so a graph with small $\ene_1(s,t)$ has a short path between $s$ and $t$.
Note also that the edge-connectivity between $s$ and $t$ is equal to the reciprocal of $\ene_{\infty}(s,t)$ (because if there are $k$ edge-disjoint $s$-$t$ paths, we can set the flow value on each path to be $1/k$), and so a graph with small $\ene_{\infty}(s,t)$ has many edge-disjoint $s$-$t$ paths.
As $\ell_2$ is between $\ell_1$ and $\ell_{\infty}$, the objective function $\Reff(s,t) = \ene^2_2(s,t)$ takes both the $s$-$t$ shortest path distance and the $s$-$t$ edge-connectivity into consideration.

A simple property suggests that $\ell_2$-energy may be even more desirable than $\ell_1$ and $\ell_{\infty}$ as a connectivity measure. 
Conceptually, adding an edge $e$ to $G$ would make $s$ and $t$ more connected. 
For $\ell_1$ and $\ell_{\infty}$, however, adding $e$ would not yield a better energy if $e$ does not improve the shortest path and the edge connectivity respectively. 
In contrast, the $\ell_2$-energy would typically improve after adding an edge, and so $\ell_2$-energy provides a smoother quantitative measure that better captures our intuition how well $s$ and $t$ are connected in a network.

Traditionally, the effective resistance has many useful probabilistic interpretations, such as the commute time~\cite{CRR+96}, the cover time~\cite{Mat88}, and the probability of an edge in a random spanning tree~\cite{Kir47}.
These interpretations suggest that the effective resistance is a useful distance function and have applications in the study of social networks.
Recently, effective resistance has found surprising applications in solving problems about graph connectivity, including constructing spectral sparsifiers~\cite{SS11} (by using the effective resistance of an edge as the sampling probability), computing maximum flow~\cite{CKM+11}, finding thin trees for ATSP~\cite{AO15}, and generating random spanning trees~\cite{MST15,Sch17}.
%


Thomson's principle also states that the electrical flow between $s$ and $t$ is the unique flow that minimizes the $\ell_2$-energy.
So, designing a network with small $s$-$t$ effective resistance has natural applications in designing electrical networks~\cite{EKPS04,GBS08,JSP12}.
One natural formulation is to keep at most $k$ wires in the input electrical network to minimize $\Reff(s,t)$, so that the electrical flow between $s$ and $t$ can still be sent with small energy while we switch off many wires in the electrical network.

Based on the above reasons, we believe that the effective resistance is a nice and natural alternative connectivity measure in network design.
More generally, it is an interesting direction to develop techniques to solve network design problems with spectral requirements.

\subsection{Main Results} \label{ss:results}

Unlike the classical problems of shortest path and min-cost flow (corresponding to the $\ell_1$ and $\ell_{\infty}$ versions of the problem), 
the $s$-$t$ effective resistance network design problem is NP-hard.

\begin{theorem} \label{t:NPc}
The $s$-$t$ effective resistance network design problem is NP-hard.
\end{theorem}

On the other hand,
we analyze a natural convex programming relaxation for the problem (Section~\ref{ss:convex}), and use it to design a constant factor approximation algorithm for the problem.

\begin{theorem} \label{t:main}
There is a convex programming based $8$-approximation randomized algorithm for the $s$-$t$ effective resistance network design problem.
\end{theorem}

The algorithm crucially uses a nice characterization of the optimal solutions to the convex program (Lemma~\ref{l:optimality}) to design a randomized path-rounding procedure (Section~\ref{ss:algorithm}) for Theorem~\ref{t:main}.

A simple example shows that the integrality gap of the convex program is at least two.
When the budget $k$ is much larger than the length of a shortest $s$-$t$ path, 
we show how to achieve an approximation ratio close to two with a randomized ``short'' path rounding algorithm (Section~\ref{ss:short}). 

\begin{theorem} \label{t:2+eps}
There is a $(2+O(\eps))$-approximation algorithm for the $s$-$t$ effective resistance network design problem, when $k \geq 2d_{st}/\eps^{10}$ where $d_{st}$ is the length of a shortest $s$-$t$ path.
\end{theorem}

\subsection{Other Results}

We consider some variants of the $s$-$t$ effective resistance network design problem, including the weighted version, the dual version, and the problem on special graphs.

There is a natural weighted generalization of the $s$-$t$ effective resistance network design problem, where we associate a cost $c_e$ and resistance $r_e$ to each edge $e$ of the input graph.

\begin{definition}[The weighted $s$-$t$ effective resistance network design problem]
\label{d:weighted}
The input is an undirected graph $G=(V,E)$ where each edge $e$ has a non-negative cost $c_e$ and a non-negative resistance $r_e$, two specified vertices $s, t \in V$, and a cost budget $k$.
The goal is to find a subgraph $H$ of $G$ that minimizes $\Reff_H(s,t)$ subject to the constraint that the total edge cost of $H$ is at most $k$. In the following, we may refer to this problem as the weighted problem for simplicity.
\end{definition}

In the weighted problem, the integrality gap of the convex program (Section~\ref{ss:convex}) becomes unbounded, even when the cost on the edges are the same ($c_e =1$ for all $e \in E$). This suggests that the weighted version may be strictly harder.
Indeed, we show stronger hardness result for the weighted problem assuming the small-set expansion conjecture~\cite{RS10,RST12}.

\begin{theorem} \label{t:SSE}
Assuming the small-set expansion conjecture, it is NP-hard to approximate the weighted $s$-$t$ effective resistance network design problem within a factor of $2 - \eps$ for any $\eps > 0$, even when $c_e=1$ for every edge $e$. 
\end{theorem}

On the other hand, when the cost on the edges are the same, the following approximation follows from the randomized path rounding algorithm in a black box manner.

\begin{corollary}
There is a convex programming based $O(R)$-approximation randomized algorithm for the weighted $s$-$t$ effective resistance network design problem when $c_e=1$ for every edge $e$, where $R = \max_e r_e/ \min_e r_e$ is the ratio between the maximum and minimum resistance.
\end{corollary}

As our problem is related to electrical network design, it is natural to consider the special case when the input graph is a series-parallel graph.
In this setting, we can use dynamic programming to design an exact algorithm for the original problem, and a fully polynomial time approximation scheme (FPTAS) for the weighted problem.

\begin{theorem} \label{t:SP}
There is an exact algorithm for the $s$-$t$ effective resistance network design problem with running time $O(|E| \cdot k^2)$ when the input graph is a series-parallel graph.

There is a $(1+\eps)$-approximation algorithm for the weighted $s$-$t$ effective resistance network design problem when the input graph is a series-parallel graph.
The running time of the algorithm is $O(|E|^7 R^2 / \eps^2)$ where $R = \max_e r_e / \min_e r_e$ is the ratio between the maximum and minimum resistance.
By a simple preprocessing scaling step, we can assume that $R$ is bounded by a polynomial, and so the algorithm is a FPTAS for the weighted problem.
\end{theorem}

We note that the integrality gap examples in Section~\ref{ss:convex} are actually series-parallel graphs, and so the dynamic programming algorithms go beyond the limitation of the natural convex program. 
We leave it as an open problem whether the weighted problem admits a constant factor approximation algorithm (possibly by combining these techniques).

We also consider the ``dual'' problem where we set the effective resistance as a hard constraint, and the objective is to minimize the number of edges in the solution subgraph.
We present similar results as the original problem in Section~\ref{ss:dual}.

\subsection{Related Work} \label{ss:previous}

In the survivable network design problem, we are given an undirected graph and a connectivity requirement $r_{u,v}$ for every pair of vertices $u,v$, and the goal is to find a minimum cost subgraph such that there are at least $r_{u,v}$ edge-disjoint paths for all $u,v$.
This problem is extensively studied and captures many interesting special cases~\cite{GGP+94,AKR95,GW95,GGT+09}.
The best approximation algorithm for this problem is due to Jain~\cite{Jai01}, who introduced the technique of iterative rounding to design a $2$-approximation algorithm.
His result has been extended in various directions, including element-connectivity~\cite{FJW01,CVV06}, directed graphs~\cite{Gab05,GGT+09}, and with degree constraints~\cite{LNSS09,EV14,FNR15,LZ15}.

Other combinatorial connectivity requirements were also considered.
A natural variation is to require $r_{u,v}$ internally vertex disjoint paths for every pair of vertices $u,v$.
This problem is much harder to approximate~\cite{KKL04,Lae14}, but there are good approximation algorithms for global connectivity~\cite{FL08,CV14} and when the maximum connectivity requirement is small~\cite{CCK08,CK09}.
Another natural problem is to require a path of length $l_{u,v}$ between every pair of vertices $u,v$.
This problem is also hard to approximate in general but there are better approximation algorithms when every edge has the same cost and the same length~\cite{DK99}.

Spectral connectivity requirements were also studied, 
including spectral gap~\cite{GB06,KMS+10} (closely related to graph expansion),
total effective resistances~\cite{GBS08}, and mixing time~\cite{BDX04}.
Some of the earlier works only proposed convex programming relaxations and heuristic algorithms.
Approximation guarantees are only obtained in two recent papers for the more general experimental design problem.
When every edge has the same cost, there is a $(1+\eps)$-approximation algorithm for minimizing the total effective resistance when the budget is at least $\Omega(|V|/\eps)$~\cite{NST18}, and there is a $(1+\eps)$-approximation algorithm for maximizing the spectral gap when the budget is at least $\Omega(|V|/\eps^2)$~\cite{ALS+17}.
For our problem, the interesting regime is when $k$ is much smaller than $|V|$, where the techniques in~\cite{ALS+17,NST18} do not apply. We have developed a set of new techniques for analyzing and rounding the solutions to the convex program that will hopefully find applications for solving related problems.




\subsection{Techniques} \label{ss:techniques}

Our main technical contribution is in designing rounding techniques for a convex programming relaxation of our problem.
There is a natural convex programming relaxation, by using the conductance of the edges as variables, and writing the $s$-$t$ effective resistance as the objective function and noting that it is convex with respect to the variables (Section~\ref{ss:convex}).

We show that optimal solutions of this convex program enjoy some nice properties\footnote{We can also show that there {\em exists} an optimal solution such that the fractional edges form a forest, but this is not included in the paper as we have not used this property in the rounding algorithm.}. Given an optimal fractional solution $x^*$ and the unit $s$-$t$ electrical flow $f^*$ supported in $x^*$, we derive from the KKT optimality conditions that there is a flow-conductance ratio $\alpha > 0$ such that $f^*_e = \alpha x^*_e$ for every fractional edge $e$ with $0 < x^*_e < 1$ and $f^*_e \geq \alpha$ for every integral edge $e$ with $x^*_e=1$.
The flow-conductance ratio $\alpha$ is crucial in the rounding algorithm and the analysis.

The rounding techniques in the two recent papers on experimental design~\cite{ALS+17,NST18} considered each edge/vector as a unit.
In~\cite{ALS+17}, a potential function as in spectral sparsification is used to guide a local search algorithm to swap two edges/vectors at a time to improve the current solution.
In~\cite{NST18}, a probability distribution on the edges/vectors is carefully designed for an independent randomized rounding.
These techniques are only known to work in the case when the solutions form a spanning set so that the ``contribution'' of each individual edge/vector is well-defined.
This is basically the reason why the results in~\cite{ALS+17,NST18} only apply when the budget $k$ is at least $\Omega(n)$.

Our approach is based on a randomized rounding procedure on $s$-$t$ paths.
Given $x^*$, we compute the unit $s$-$t$ electrical flow  $f^*$ supported in $x^*$,
and decompose $f^*$ as a convex combination of $s$-$t$ paths.
The rounding algorithm has $T=1/\alpha$ iterations (recall that $\alpha$ is the flow-conductance ratio of the optimal solution $x^*$),
where we pick a random path $P_i$ from the convex combination in each iteration, and return $H := \cup_{i=1}^T P_i$ as our solution. 
One difference from the previous techniques in the literature is that each unit in the rounding algorithm is a $s$-$t$ path, so in particular $s$ and $t$ are always connected in our solution. 
Another difference is that our problem has some extra structure, so that we can compute the electrical flow $f^*$ to guide our rounding procedure, where the variables $f^*_e$ are not in the convex program.
These allow us to obtain a constant factor approximation algorithm for all budget $k \geq d_{st}$ (note that when $k<d_{st}$ there is no feasible integral solution).

In the analysis, we prove in Lemma~\ref{l:budget} that the expected number of edges in $H$ is at most $k$, and in Lemma~\ref{l:energy} that the expected effective resistance is $\Reff_H(s,t) \leq 2\Reff_{x^*}(s,t)$.
To bound the expected effective resistance, we use Thomson's principle and construct a unit $s$-$t$ flow $F$ to show that $\Reff_H(s,t) \leq \ene_H(F) \leq 2\Reff_{x^*}(s,t)$.
To construct the unit $s$-$t$ flow $F$, we keep the flow-conductance ratio and send $\alpha$ units of flow on each sampled path $P_i$ (i.e. $f_e=\alpha$ and $x_e=1$).
The flow-conductance ratio plays a crucial role in the proofs of both lemmas.
This is because the rounding algorithm is based on the flow variables $f^*_e$, and thus the performance guarantees are in terms of $f^*_e$, but the ratio $\alpha$ allows us to relate them back to the variables $x^*_e$ in the convex program.
Combining the two lemmas give us a constant factor bicriteria approximation algorithm for the problem.
This can be turned into a true approximation algorithm by scaling down the budget to $k/2$ and run the bicriteria approximation algorithm with some additional claims (Section~\ref{ss:constant}).

The improvement on the approximation ratio when budget $k$ is large comes from two observations. 
The first is that if $k$ is much larger than the length of the shortest $s$-$t$ path, then the number of independent iterations in the rounding scheme is large (Lemma~\ref{l:alpha}). 
The second is that we can ignore some $s$-$t$ paths in the flow decomposition with many fractional edges without affecting the performance much.
Combining these, we can apply a Chernoff-Hoeffding bound to show that the number of edges is at most $(1+\eps)k$ with high probability.
Then it is not necessary to scale down the budget by a factor of $2$ and we can prove a stronger bound that the effective resistance is at most $2+O(\eps)$ times the optimal value.

\subsection{Organization}

In Section~\ref{s:background}, we define the notations used in this paper and cover background knowledge on effective resistances.
We present the convex programming relaxation and our two rounding procedures in Section~\ref{s:rounding}, and the dynamic programming algorithm in Section~\ref{s:dynamic}.
The NP-hardness and small set expansion hardness results are provided in Section~\ref{s:hardness}. 

\section{Preliminaries} \label{s:background}

We introduce the notations and definitions for graphs and matrices in Section~\ref{ss:graph}, and then define electrical flow and effective resistance and state some basic results in Section~\ref{ss:electric}.

\subsection{Graphs and Matrices} \label{ss:graph}

Let $G=(V,E)$ be an undirected graph with edge weight $w_e\geq 0$ on each edge $e \in E$.
The number of vertices and the number of edges are denoted by $n:=|V|$ and $m:=|E|$.
For a subset of edges $F \subseteq E$, the total weight of edges in $F$ is $w(F) := \sum_{e \in F} w_e$.
For a subset of vertices $S \subseteq V$, the set of edges with one endpoint in $S$ and one endpoint in $V-S$ is denoted by $\delta(S)$.
For a vertex $v$, the set of edges incident on a vertex $v$ is $\delta(v):=\delta(\{v\})$, and the weighted degree of $v$ is $\deg(v) := w(\delta(v))$.
The volume of a set $\vol(S) := \sum_{v \in S} \deg(v)$ is defined as the sum of the weighted degrees of vertices in $S$.
The conductance of a set $\phi(S) := w(\delta(S))/\vol(S)$ is defined as the ratio of the total weight on the boundary of $S$ to the total weighted degrees in $S$.
For two subsets $S_1,S_2 \subseteq V$, the set of edges with one endpoint in $S_1$ and one endpoint in $S_2$ is denoted by $E(S_1,S_2)$.

In this paper, an undirected graph $G=(V,E)$ with non-negative edge weights $w \in \R^E$ is interpreted as an electrical network, where each edge $e$ is a resistor with conductance $w_e$ (not to be confused with the conductance $\phi(S)$ of a set $S$ as defined above), or equivalently with resistance $r_e := 1/w_e$. 
The adjacency matrix $A \in \R^{V \times V}$ of the graph is defined as $A_{u,v} = w_{u,v}$ for all $u,v \in V$.
The Laplacian matrix $L \in \R^{V \times V}$ of the graph is defined as $L = D - A$ where $D \in \R^{V \times V}$ is the diagonal degree matrix with $D_{u,u} = \deg(u)$ for all $u \in V$.
For each edge $e = uv \in E$, let $b_e := \chi_u - \chi_v$ where $\chi_u \in \R^n$ is the vector with one in the $u$-th entry and zero otherwise. 
The Laplacian matrix can also be written as 
\[
L = \sum_{e \in E} w_e b_e b_e^T.
\]
Let $\lambda_1 \leq \lambda_2 \leq \ldots \leq \lambda_n$ be the eigenvalues of $L$ with corresponding orthonormal eigenvectors $v_1, v_2, \ldots, v_n$
so that $L = \sum_{i=1}^n \lambda_i v_i v_i^T$.
It is well-known that the Laplacian matrix is positive semidefinite and $\lambda_1=0$ with $v_1 = \vec{1}/\sqrt{n}$ as the corresponding eigenvector,
and $\lambda_2 > 0$ if and only if $G$ is connected.
The pseudo-inverse of the Laplacian matrix $L$ of a connected graph is defined as 
\[L^\dagger = \sum_{i=2}^n \frac{1}{\lambda_i} v_i v_i^T,\]
which maps every vector $x$ orthogonal to $v_1$ to a vector $y$ such that $Ly=x$.

\subsection{Electrical Flow and Effective Resistance} \label{ss:electric}

Before defining $s$-$t$ electrical flow,
we first define the standard unit $s$-$t$ flow.
For each edge $e=uv$, we have two variables $f(uv)$ and $f(vu)$ with $f(uv) = -f(vu)$, where $f(uv)$ is positive if the flow is going from $u$ to $v$ and negative otherwise.
A unit $s$-$t$ flow $f$ satisfies the following flow conservation constraints:
\[ \sum_{u \in \delta(v)} f(vu) =
  \begin{cases}
    1 & v = s\\
    -1 & v = t\\
    0 & \text{otherwise.} 
  \end{cases}\]
Given a unit $s$-$t$ flow $f$, we overload the notation and define its undirected flow vector $f: E \to \R_{\geq 0}$ with $f_e := |f(uv)|$ for each edge $e = uv$.
A unit $s$-$t$ electrical flow is a unit $s$-$t$ flow $f$ that also satisfies the Ohm's law:
There exists a potential vector $\varphi \in \R^V$ such that for all $u,v \in V$,
\begin{equation*}
f(uv) = w_{uv} \cdot (\varphi(u) - \varphi(v)). 
\end{equation*}
The effective resistance between $s$ and $t$ is defined as 
\begin{equation*}
\Reff(s,t) := \varphi(s)-\varphi(t),  
\end{equation*}
which is the potential difference between $s$ and $t$ when one unit of electrical flow is sent from $s$ to $t$.
The $s$-$t$ effective resistance can be interpreted as the resistance of the whole graph $G$ as a big resistor when an electrical flow is sent from $s$ to $t$.

One can write the effective resistance in terms of the Laplacian matrix.
For $u,v \in V$, let $b_{uv} = \ch_u - \ch_v$, where $\ch_v \in \R^n$ is the unit vector with $1$ in the $v$-th entry and $0$ in other entries. 
Combining the flow conservation constraint and the Ohm's law, it can be checked that the potential vector $\varphi \in \R^V$ of a unit $s$-$t$ electrical flow is a solution to the linear system
\[L \cdot \varphi = b_{st}.\]
Note that $\varphi = L^\dagger b_{st}$ is a solution, and if $G$ is connected then any solution is given by $p+c\cdot \vec 1$ for $c\in\mathbb{R}$.
Therefore, we can write
\[\Reff(s,t) = \varphi(s) - \varphi(t) = b_{st}^T L^\dagger b_{st}.\]

The effective resistance can also be characterized by the energy of a flow.
The energy 
of an $s$-$t$ flow $f$ is defined as
\[\ene(f) := \sum_{e \in E} \frac{f_e^2}{w_e} = \sum_{e \in E} r_e f_e^2.\]
Thomson's principle~\cite{KT67} states that the unit $s$-$t$ electrical flow is the unique unit $s$-$t$ flow that minimizes the energy.
This can be verified by writing down the optimality condition of the minimization problem.
Moreover, this energy is exactly the $s$-$t$ effective resistance.
To see this, note that the flow value on edge $uv$ in the unit $s$-$t$ electrical flow satisfies $f(uv) = w_{uv} \cdot (\varphi(u)-\varphi(v)) = w_{uv} \cdot b^T_{uv} L^\dagger b_{st}$ and thus
\begin{equation*}
\ene(f) 
= \sum_{uv \in E} w_{uv} (b^T_{uv} L^\dagger b_{st})^2
= b_{st}^T L^\dagger \left(\sum_{uv \in E} w_{uv} b_{uv} b_{uv}^T \right) L^\dagger b_{st}
= b_{st}^T L^\dagger L L^\dagger b_{st} 
= \Reff(s,t).
\end{equation*}
To summarize, we will use the following result from Thomson's principle.
\begin{fact}[Thomson's principle~\cite{KT67}] \label{f:Thomson}
Let $f^*$ be the unit electrical $s$-$t$ flow in $G$.  Then
\[{\rm Reff}_G(s,t) = \min_f \{ \ene(f)~|~f {\rm~is~a~unit~} s\text{-}t {\rm~flow~in~} G\} =
\ene(f^*).\] 
\end{fact}
A corollary of Thomson's principle is the following intuitive result known as the Rayleigh's monotonicity principle.
\begin{fact}[Rayleigh's monotonicity principle] \label{f:monotonicity}
The $s$-$t$ effective resistance cannot increase if the resistance of an edge is decreased.
\end{fact}
We will also use the following result to write a convex programming relaxation of our problem.
\begin{fact}[\cite{GBS08}] \label{f:convex}
The $s$-$t$ effective resistance is a convex function with respect to the conductance of the edges.
\end{fact}



\section{Convex Programming Algorithm} \label{s:rounding}

In this section, we analyze a convex programming relaxation for our problem.
We first describe the convex program and prove a characterization of the optimal solutions in Section~\ref{ss:convex}.
We then present a randomized rounding algorithm using flow decomposition in Section~\ref{ss:algorithm}, and show that it is a constant factor bicriteria approximation algorithm in Section~\ref{ss:bicriteria}.
Then, we show how to convert the bicriteria approximation algorithm into a true approximation algorithm in Section~\ref{ss:constant}, and how to modify the algorithm slightly to achieve a better approximation guarantee when the budget $k$ is large in Section~\ref{ss:short}.
Finally, we discuss the dual problem of minimizing the cost while satisfying the effective resistance constraint in Section~\ref{ss:dual}.

\subsection{Convex Programming Relaxation} \label{ss:convex}

The formulation is for the weighted problem,
where each edge has a weight $w_e := 1/r_e$.
We introduce a variable $x_e$ for each edge $e$ to indicate whether $e$ is chosen in our subgraph. 
Let $$L_x := \sum_{e \in E} x_e w_e b_e b_e^T$$ be the Laplacian matrix of the fractional solution $x$, and $\Reff_x(s,t)$ be the $s$-$t$ effective resistance of the graph with conductance $x_e w_e$ on edge $e\in E$. 
The following is a natural convex programming relaxation for the problem.

\begin{equation}
\begin{aligned}
& \underset{x \in \R^m}{\min} 
& & \Reff_x(s,t) = b_{st}^T L_x^\dagger b_{st}\\
& \text{subject to} & & \sum_{e \in E} c_e  x_e \leq k, \\
&&& 0 \leq x_e \leq 1, \qquad \forall e \in E.
\end{aligned} \tag{CP} \label{eq:P}
\end{equation}

This is an exact formulation if $x_e \in \{0,1\}$ for all $e \in E$.
The objective function is convex in $x$ by Fact~\ref{f:convex}.
The convex program can be solved in polynomial time by the ellipsoid method to inverse exponential accuracy, or by the techniques described in~\cite{ALS+17} to inverse polynomial accuracy, which are both sufficient for the rounding algorithm.

\subsubsection{Integrality Gap Examples} \label{ss:gap}

We show some limitations of the convex program for general $w_e$ and $c_e$.
The following figure shows a simple example where the integrality gap is unbounded if the cost could be arbitrary.

\begin{figure}[!ht]
\centering
\resizebox{0.65\textwidth}{!}{
\begin{tikzpicture}
	\tikzset{VertexStyle/.style = {shape = circle, fill = black, minimum size = 0.2}}
	\Vertex[Math,x=0,y=0,LabelOut,Lpos=180]{s}
	\Vertex[x=10,y=0,LabelOut]{t}	
	
    \Vertex[x=2,y=1,NoLabel]{v1}
    \Vertex[x=4,y=1,NoLabel]{v2}
    \node (dots) at (5, 1) {$\dots$};
    \Vertex[x=6,y=1,NoLabel]{v3}
    \Vertex[x=8,y=1,NoLabel]{v4}
    
    \Vertex[x=5,y=0,NoLabel]{v5}
    
    \Edge(v1)(v2)
    \Edge(v3)(v4)
	
	\Edge(t)(v5)
    \Edge(v5)(s)
    
    \Edge(s)(v1)
    \Edge(v4)(t)
\end{tikzpicture}
}
\caption{Integrality gap example with arbitrary cost and unit resistance.} 
\label{fig:intgapACUR}
\end{figure}
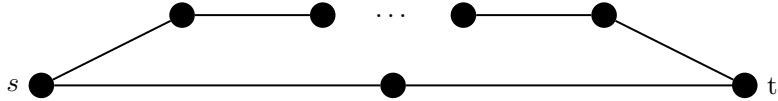

In this graph, the top path has length $n\!-\!2$ where each edge has cost $1/(n\!-\!2)$.
The bottom path has two edges with cost $1$. 
The resistance of each edge is $1$, and the budget is $k = 1$. 
The integrality gap of this example is $\Omega(n)$.
To see this, the integral solution can only afford the top path,
and the effective resistance is $n\!-\!2$.
However, the fractional solution can set $x_e=1/2$ for each of the two bottom edges, and the effective resistance of this fractional solution is $4$.

The following figure shows another simple example where the integrality gap is unbounded if the edge costs are the same but the resistances could be arbitrary.

\begin{figure}[!ht]
\centering
\resizebox{0.65\textwidth}{!}{
\begin{tikzpicture}
	\tikzset{VertexStyle/.style = {shape = circle, fill = black, minimum size = 0.2}}
	\Vertex[Math,x=0,y=0,LabelOut,Lpos=180]{s}
	\Vertex[x=10,y=0,LabelOut]{t}	
	
    \Vertex[x=2,y=1,NoLabel]{v1}
    \Vertex[x=4,y=1,NoLabel]{v2}
    \node (dots) at (5, 1) {$\dots$};
    \Vertex[x=6,y=1,NoLabel]{v3}
    \Vertex[x=8,y=1,NoLabel]{v4}
    
    \Edge(v1)(v2)
    \Edge(v3)(v4)
	\Edge(t)(s)	
    
    \Edge(s)(v1)
    \Edge(v4)(t)
	
\end{tikzpicture}
}
\caption{Integrality gap example with arbitrary resistance and unit cost.}

\end{figure}
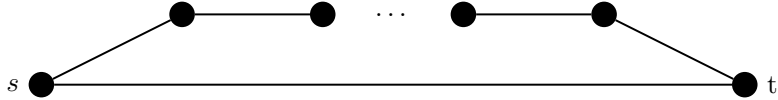

In this example, the top path has length $n\!-\!1$ with each edge of resistance $1$.
The bottom path has only one edge with resistance $R$. 
All edges have cost $1$ and the budget $k = n\!-\!2$.
The integral solution can only afford the bottom path, 
with effective resistance $R$. 
The fractional solution can set $x_e=(n\!-\!2)/(n\!-\!1)$ for each edge in the top path,
with effective resistance $O(n)$.
When $R \gg n$, the integrality gap could be arbitrarily large.

Even in the unit-cost unit-resistance case, the integrality gap is unbounded if $k$ is smaller than the $s$-$t$ shortest path distance.
Henceforth, in view of these observations we assume the following in the rest of this section.

\begin{assumption} \label{a:assumption}
We assume that $c_e=w_e=r_e=1$ for every edge $e \in E$, which is the setting of the $s$-$t$ effective resistance network design problem,
and the budget $k$ is at least the shortest path distance $d_{st}$ between $s$ and $t$ in the input graph.
\end{assumption} 

The integrality gap of the convex program is still at least two with Assumption~\ref{a:assumption}.
For a simple example, consider a graph with two vertex-disjoint $s$-$t$ paths, each of length $k/2 + 1$, and the budget is $k$.
Then the optimal integral value is $k/2 + 1$ while the optimal fractional value is close to $k/4$, and so the integrality gap gets arbitrarily close to two.

We will show that the integrality gap of the convex program is at most $8$ with these assumptions.
Note that just to connect $s$ and $t$, then $k$ must be at least the $s$-$t$ shortest path distance.
It is interesting that this small additional assumption could reduce the integrality gap from unbounded to a constant.

%
%
%

\subsubsection{Characterization of Optimal Solutions} \label{ss:optimality}

In the case $c_e=w_e=r_e=1$ for all edges $e \in E$, we will prove that the electrical flow $f^*$ supported in the optimal solution $x^*$ to~(\ref{eq:P}) satisfies a crucial property about the flow-conductance ratio $f^*_e/x^*_e$.

\begin{lemma}[Characterization of Optimal Solution] \label{l:optimality}
Let $G=(V,E)$ be the input graph with $c_e=w_e=1$ for all edges $e \in E$.
Let $x^*:E \to \R_{\geq 0}$ be an optimal solution to the convex program~\eqref{eq:P}.
Let $E_F \subseteq E$ be the set of fractional edges with $0 < x^*_e < 1$, 
and $E_I \subseteq E$ be the set of integral edges with $x^*_e = 1$.
Let $f^*: E \to \R_{\geq 0}$ be the undirected flow vector of the unit $s$-$t$ electrical flow supported in $x^*$.
There exists $\alpha > 0$ such that
\[
f^*_{e} = \alpha x^*_{e} ~~~\forall e \in E_F
\quad {\rm and} \quad
f^*_{e} \geq \alpha ~~~\forall e \in E_I.
\]
\end{lemma}

\begin{proof}
By removing edges with $x^*_e=0$, we can assume $x^*_e > 0$ for every $e \in E$.
By removing isolated vertices, we can further assume that the nonzero edges form a connected graph.
So, we can write $\Reff_{x^*}(s,t) = b_{st}^T L_{x^*}^\dagger b_{st}$, 
where $L_{x^*}$ has rank $n-1$ and the null space of $L_{x^*}$ is $\sspan(\vec{1})$. 
Since $b_{st} \perp \vec{1}$, we have $L_{x^*} L_{x^*}^\dagger b_{st} = b_{st}$ and $L_{x^*}^\dagger b_{st} \perp \vec{1}$, which implies that $L_{x^*}^\dagger b_{st} = (L_{x^*} + \frac{1}{n}J)^{-1} b_{st}$ where $J$ is the all-ones matrix. 
Using the fact that $\partial A^{-1} = -A^{-1} (\partial A) A^{-1}$ (see e.g.~\cite{PP12}), 
we derive
\begin{eqnarray*}
\nabla_{x^*_e} \Reff_{x^*}(s,t) 
& = & \nabla_{x^*_e} b_{st}^T \left(L_{x^*} + \frac{1}{n} J\right)^{-1} b_{st} 
~ = ~ - b_{st}^T \left(L_{x^*} + \frac{1}{n} J\right)^{-1} \left(\nabla_{x^*_e} L_{x^*} \right) \left(L_{x^*} + \frac{1}{n} J\right)^{-1} b_{st} 
\\
& = & -b_{st}^T L_{x^*}^\dagger \left(\nabla_{x^*_e} \sum_{e \in E} x^*_e w_e b_e b_e^T \right) L_{x^*}^\dagger b_{st} 
~ = ~ -b_{st}^T L_{x^*}^\dagger b_e b_e^T L_{x^*}^\dagger b_{st} 
~ = ~ -(b_{st}^T L_{x^*}^\dagger b_e)^2,
\end{eqnarray*}
where we used the assumption that $w_e=1$ for all $e \in E$.
With this, we write down the KKT conditions for the convex program. 
Let $\mu$ be the dual variable for the budget constraint $\sum_{e \in E} c_e x^*_e \leq k$, and $\lambda^+_e$ and $\lambda^-_e$ be the dual variables for the upper bound $x^*_e \leq 1$ and the nonnegative constraint $x^*_e \geq 0$ respectively.
The KKT conditions states if $x^*$ is an optimal solution to \eqref{eq:P}, then there exist $\lambda^+, \lambda^-$ and $\mu$ such that
\begin{equation*}
\begin{aligned}
	& \sum_{e \in E} x^*_e \leq k, \quad 0 \leq x^*_e \leq 1~~\forall e \in E, & & \text{~~(Primal feasibility)} \\
        & \mu \geq 0, \quad \lambda^+_e \geq 0 {\rm~and~} \lambda^-_e \geq 0~~\forall e \in E, & & \text{~~(Dual feasibility)} \\
        & \mu \cdot \left(k - \sum_{e \in E} x^*_e \right) = 0, \quad 
	  \lambda^+_e \cdot (x^*_e - 1) = 0 {\rm~and~}
          \lambda^-_e \cdot x^*_e = 0~~\forall e \in E,   & & \text{~~(Complementary slackness)} \\
        & (b_e^T L_{x^*}^\dagger b_{st})^2 = \lambda^+_e - \lambda^-_e + c_e\mu = \lambda^+_e - \lambda^-_e + \mu, & & \text{~~(Lagrangian optimality)}
\end{aligned}
\end{equation*}
where we used the assumption that $c_e=1$ for all $e \in E$.
For an integral edge with $x^*_e = 1$, we have $\lambda^-_e = 0$ by the complementary slackness condition.
Since $\lambda^+_e \geq 0$, it follows from the Lagrangian optimality condition that $(b_e^T L_{x^*}^\dagger b_{st})^2 \geq \mu$.
For a fractional edge with $0 < x^*_e < 1$, we have $\lambda^+_e = \lambda^-_e = 0$ by the complementary slackness condition, and therefore $(b_e^T L_{x^*}^\dagger b_{st})^2 = \mu$ by the Lagrangian optimality condition.
We can assume that $\mu > 0$.
Otherwise, $\mu = 0$ implies that the flow on all fractional edges are zero, and so we can delete them from the graph without affecting the $s$-$t$ effective resistance, and we have an integral solution.

Let $\varphi$ be a potential vector of the electrical flow $f^*$ supported in $x^*$.
For an edge $e=uv \in E$,
\[\Big(\frac{f^*_{e}}{x^*_e}\Big)^2 = (\varphi(u) - \varphi(v))^2 = \left(b_{e}^T L_{x^*}^\dagger b_{st}\right)^2,\]
where the first equality is by Ohm's law and the assumption that $w_{uv}=1$ for all $uv \in E$, and the second equality uses that $L_{x^*} \varphi = b_{st}$ as explained in Section~\ref{ss:electric}.
The lemma then follows from the above paragraph and writing $\mu$ as $\alpha^2$.
\end{proof}

The flow-conductance ratio $\alpha$ will be crucial in the rounding algorithm and its analysis.
The following lemma shows an upper bound on $\alpha$ using the budget $k$ and the shortest path distance $d_{st}$ between $s$ and $t$.

\begin{lemma} \label{l:alpha}
Under the conditions in Assumption~\ref{a:assumption}, it holds that $\alpha^2 \leq d_{st}/k \leq 1$.
\end{lemma}
\begin{proof}
Let $x^*$ be an optimal solution to \eqref{eq:P}, 
and $f^*$ be the unit $s$-$t$ electrical flow supported in $x^*$.
As $k \geq d_{st}$, a shortest path is a feasible solution to \eqref{eq:P}, and thus $\Reff_{x^*}(s,t) \leq d_{st}$. 
On the other hand, by Thomson's principle and Lemma~\ref{l:optimality},
\[
\Reff_{x^*}(s,t) 
= \sum_{e \in E} \frac{(f^*_e)^2}{x^*_e} 
= \sum_{e \in E_I} (f^*_e)^2 + \sum_{e \in E_F} \frac{(f^*_e)^2}{x^*_e} 
\geq \sum_{e \in E_I} \alpha^2 + \sum_{e \in E_F} \alpha^2 x^*_e 
= \alpha^2 \sum_{e \in E} x^*_e 
= \alpha^2 k,
\]
where the last equality holds since we can assume $\sum_{e \in E} x^*_e = k$ for the optimal solution $x^*$ without loss of generality by Rayleigh's principle (or otherwise we have an integral optimal solution).
The lemma follows by combining the upper bound and the lower bound.
\end{proof}



\subsection{Randomized Path-Rounding Algorithm} \label{ss:algorithm}

Our rounding algorithm uses the unit electrical flow $f^*$ supported in the optimal solution $x^*$ to construct an integral solution.
The algorithm will first decompose the flow $f^*$ as a convex combination of flow paths, and then randomly choose the flow paths and return the union of the chosen flow paths as our solution.

The following lemma about flow decomposition is by the standard argument to remove one (fractional) flow path at a time, which holds for any unit directed acyclic $s$-$t$ flow.

\begin{lemma}[Flow Decomposition] \label{l:decomposition}
Given a unit $s$-$t$ electrical flow $f$,
there is a polynomial time algorithm to find a set ${\mathcal P}$ of $s$-$t$ paths with $|{\mathcal P}| \leq |E|$ such that
the undirected flow vector $f : E \to \R_{\geq 0}$ can be written as a convex combination of the characteristic vectors of the paths in ${\mathcal P}$, i.e.
\[
f = \sum_{p \in {\mathcal P}} v_p \cdot \chi_{p} 
\quad {\rm and} \quad
\sum_{p \in {\mathcal P}} v_p = 1
\quad {\rm and} \quad
v_p > 0~{\rm for~each~} p \in {\mathcal P},
\]
where $\chi_{p} \in R^{|E|}$ is the characteristic vector of the path $p$ with one on each edge $e \in p$ and zero otherwise.
\end{lemma}

With the flow decomposition, we are ready to present the rounding algorithm.

\begin{framed}
{\bf Randomized Path Rounding Algorithm}
\begin{enumerate}
\item Let $x^*$ be an optimal solution to the convex program \eqref{eq:P}.
Let $f^*$ be the unit $s$-$t$ electrical flow supported in $x^*$.
Let $\alpha$ be the flow-conductance ratio defined in Lemma~\ref{l:optimality}.

\item Compute a flow decomposition ${\mathcal P}$ of $f^*$ as defined in Lemma~\ref{l:decomposition}.


  
\item For $i$ from $1$ to $T:=\lfloor 1/\alpha \rfloor$ do
  \begin{itemize}
     \item Let $P_i$ be a random path from ${\mathcal P}$ where each path $p \in {\mathcal P}$ is sampled with probability $v_p$.
  \end{itemize}

\item Return the subgraph $H$ formed by the edge set $\cup_{i=1}^T P_i$.
\end{enumerate}
\end{framed}

The following lemma shows that the rounding algorithm will always return a non-empty subgraph.

\begin{lemma} \label{l:T}
Suppose the input instance satisfies the conditions in Assumption~\ref{a:assumption}.
Let $x^*$ be an optimal solution to \eqref{eq:P} and $\alpha>0$ be the flow-conductance ratio as defined in Lemma~\ref{l:optimality}.
Then 
\[\frac{1}{\alpha} \geq T \geq \frac{1}{2\alpha} > 0.\] 
\end{lemma}
\begin{proof}
Since we assumed that the budget $k$ is at least the length $d_{st}$ of a shortest $s$-$t$ path, it follows from Lemma~\ref{l:alpha} that $\alpha \leq 1$.
This implies that
\[
\frac{1}{\alpha} \geq T 
= \left\lfloor \frac{1}{\alpha} \right\rfloor 
\geq \max \left\{1,\frac{1}{\alpha} -1\right\}
\quad \implies \quad
1 \geq T\alpha \geq \max\{\alpha,1-\alpha\} \geq \frac{1}{2}.
\]
\end{proof}


\subsection{Bicriteria Approximation} \label{ss:bicriteria}

The analysis of the approximation guarantee goes as follows.
First, we show that the expected number of edge in the returned subgraph $H$ is at most the budget $k$.
Then, we prove that the expected effective resistance of the returned subgraph is at most two times that of the optimal fractional solution.
Both of these steps use the flow-conductance ratio $\alpha$ crucially.
These combine to show that the randomized path rounding algorithm is a constant factor bicriteria approximation algorithm.

Let $x^*$ be an optimal solution to \eqref{eq:P}.
Let $E_F$ and $E_I$ be the set of fractional edges and integral edges in $x^*$.
We assume that each edge $e \in E_I$ will be included in the subgraph $H$ returned by the rounding algorithm.
We focus on bounding the number of edges in $E_F$ that will be included in $H$.

\begin{lemma}[Expected Budget] \label{l:budget}
Let $x^*$ be an optimal solution to \eqref{eq:P} when $w_e=1$ for all edges $e \in E$.
Let $X_e$ be an indicator variable of whether $e$ is included in the returned subgraph $H$ by the rounding algorithm, 
Then,
\[
\E\left[ \sum_{e \in E_F} X_e \right] 
\leq T \alpha \sum_{e \in E_F} x^*_e
\leq \sum_{e \in E_F} x^*_e.
\]
\end{lemma}
\begin{proof}
Note that an edge $e$ is contained in $P_i$ with probability $\sum_{p \in {\mathcal P}: p \ni e} v_p$.
By the union bound, an edge $e$ is included in the returned subgraph $H$ by the rounding algorithm with probability
\[
\prob(X_e = 1) 
\leq \sum_{i=1}^{T} \sum_{p \in {\mathcal P}:p \ni e} v_p 
= T \sum_{p \in {\mathcal P}:p \ni e} v_p 
= T f^*_e,
\]
where the last equality holds by the property of the flow decomposition ${\mathcal P}$ of the electrical flow $f^*$ in Lemma~\ref{l:decomposition}.

By Lemma~\ref{l:optimality}, $f^*_e = \alpha x^*_e$ for each fractional edge $e \in E_F$, and this implies that
\[
\prob(X_e = 1) \leq Tf_e^* = T\alpha x^*_e~~~\forall e \in E_F.
\]
Therefore, 
\[
\E\left[ \sum_{e \in E_F} X_e \right] 
= \sum_{e \in E_F} \prob(X_e = 1) 
\leq T\alpha \sum_{e \in E_F} x^*_e
= \left\lfloor \frac{1}{\alpha} \right\rfloor \alpha \sum_{e \in E_F} x^*_e
\leq \sum_{e \in E_F} x^*_e.
\]
%
\end{proof}

The key step is to show that $\E[\Reff_H(s,t)] \leq 2\Reff_{x^*}(s,t)$.
To prove this, we construct a unit $s$-$t$ flow $F$ and show that $\E[\ene_H(F)] \leq 2\Reff_{x^*}(s,t)$, and hence by Thomson's principle $\E[\Reff_H(s,t)] \leq \E[\ene_H(F)] \leq 2\Reff_{x^*}(s,t)$.
To construct the flow $F$, the idea is to follow the ratio $\alpha$ in the fractional solution $x^*$ and send $\alpha$ units of flow on each path $P_i$ selected.

\begin{lemma}[Expected Effective Resistance] \label{l:energy}
Suppose the input instance satisfies the conditions in Assumption~\ref{a:assumption}.
Let $x^*$ be an optimal solution to \eqref{eq:P} 
and $f^*$ be the unit $s$-$t$ electrical flow supported in $x^*$.
The expected $s$-$t$ effective resistance of the subgraph $H$ returned by the rounding algorithm is 
\[
\E\left[{\rm Reff}_H(s,t)\right] 
\leq \left(1 - \frac{1}{T} + \frac{1}{T\alpha} \right) \cdot \ene_{x^*}(f^*) 
= \left(1 - \frac{1}{T} + \frac{1}{T\alpha} \right) \cdot {\rm Reff}_{x^*}(s,t) 
\leq 2 {\rm Reff}_{x^*}(s,t). 
\]
\end{lemma}
\begin{proof}
Consider the undirected flow vector $F: E \to \R_{\geq 0}$ defined by sending $\alpha$ units of flow on each path $P_i$ chosen by the rounding algorithm, i.e. the random variable $F = \sum_{i=1}^T \alpha \cdot \chi_{P_i}$ with $F_e = \alpha \cdot \#\{P_i~|~1 \leq i \leq T, P_i \ni e\}$ for each edge $e \in E$.
We would like to upper bound the expected energy $\ene_H(F)$ in order to upper bound $\Reff_H(s,t)$.
 
Each $P_i$ is a random $s$-$t$ path sampled from the flow decomposition ${\mathcal P}$ of the undirected flow vector $f^* : E \to \R_{\geq 0}$ of the unit $s$-$t$ electrical flow supported in $x^*$, and $\chi_{P_i} \in \R^m$ is its characteristic vector with expected value
\[
\E[\chi_{P_i}] = \sum_{p \in {\mathcal P}} v_p \cdot \chi_p = f^*.
\]
Since each edge in $H$ is of conductance one,
the expected energy of $F$ in $H$ is
\[
\E\left[ \ene_H(F) \right] 
= \E \left[ \sum_{e \in E} F_e^2 \right]
= \E [\langle F, F \rangle]
= \E \left[\left \langle \sum_{i=1}^T \alpha \cdot \chi_{P_i}, \sum_{j=1}^T \alpha \cdot \chi_{P_j} \right \rangle \right]
= \sum_{i=1}^T \sum_{j=1}^T \alpha^2 \cdot \E [\langle \chi_{P_i}, \chi_{P_j} \rangle].
\]
As each path $P_i$ is sampled independently, for $i \neq j$,
\[
\E[\langle \chi_{P_i}, \chi_{P_j} \rangle]
= \langle \E[\chi_{P_i}], \E[\chi_{P_j}] \rangle
= \langle f^*, f^*\rangle
= \sum_{e \in E} (f^*_e)^2.
\]
For $i = j$,
\[
\E[\langle \chi_{P_i}, \chi_{P_i} \rangle]
= \sum_{p \in {\mathcal P}} v_p \langle \chi_p, \chi_p \rangle
= \sum_{p \in {\mathcal P}} v_p \sum_{e \in p} 1
= \sum_{e \in E}~ \sum_{p \in {\mathcal P}: p \ni e} v_p
= \sum_{e \in E} f^*_e,
\]
where the last equality follows from the property of the flow decomposition in Lemma~\ref{l:decomposition}.
Combining these two terms, it follows that
\[
\E\left[ \ene_H(F) \right] 
= \alpha^2 T \sum_{e \in E} f_e^* + \alpha^2 T(T-1) \sum_{e \in E} (f^*_e)^2.
\]
Thomson's principle states that the $\Reff_H(s,t)$ is upper bounded by the energy of any one unit $s$-$t$ flow.
Note that $F$ is an $s$-$t$ flow of $T \alpha$ units, and $T \alpha > 0$ by Lemma~\ref{l:T}.
Scaling $F$ to a one unit $s$-$t$ flow by dividing the flow on each edge by $T \alpha$ gives an upper bound on
\begin{eqnarray*}
\E[\Reff_H(s,t)] 
\leq \frac{\E\left[ \ene_H(F) \right]}{T^2 \alpha^2}
& = & \frac{1}{T} \sum_{e \in E} f_e^* + \Big(1-\frac{1}{T}\Big) \sum_{e \in E} (f^*_e)^2
\\
& \leq & \frac{1}{T \alpha} \sum_{e \in E} \frac{(f^*_e)^2}{x^*_e} + \Big(1-\frac{1}{T}\Big) \sum_{e \in E} \frac{(f^*_e)^2}{x^*_e}
\\
& = & \left(1-\frac{1}{T}+\frac{1}{T\alpha}\right) \cdot \ene_{x^*}(f^*)
\\
& = & \left(1-\frac{1}{T}+\frac{1}{T\alpha}\right) \cdot \Reff_{x^*}(s,t),
\end{eqnarray*}
where the second inequality follows from Lemma~\ref{l:optimality} that $f^*_e/x^*_e \geq \alpha$ for every edge $e \in E$ and also $x^*_e \leq 1$ for every edge $e \in E$, and the last equality is from Thomson's principle that $\Reff_{x^*}(s,t) = \ene_{x^*}(f^*)$.
Finally, notice that $1- 1/T + 1/(T\alpha) \leq 2$ as $1/\alpha - 1 \leq \lfloor 1/\alpha \rfloor = T$ .
\end{proof}

Combining Lemma~\ref{l:budget} and Lemma~\ref{l:energy}, it follows from a simple application of Markov's inequality that there is an outcome of the randomized path-rounding algorithm which uses at most $2k$ edges with $s$-$t$ effective resistance at most $4\Reff_{x^*}(s,t)$.
In the following, we apply Markov's inequality more carefully to show that the success probability is at least $\Omega(\alpha)$.
In the next subsection, we will argue that $\alpha$ can be assumed to be $\Omega(1/m)$ and so the path-rounding algorithm is a randomized polynomial time algorithm.

\begin{theorem}[Bicriteria Approximation] \label{t:bicriteria}
Suppose the input instance satisfies the conditions in Assumption~\ref{a:assumption}.
Let $x^*$ be an optimal solution to \eqref{eq:P}.
Given $x^*$, the randomized path rounding algorithm will return a subgraph $H$ with at most $2k$ edges and $\Reff_H(s,t) \leq 4\Reff_{x^*}(s,t)$ with probability at least $\Omega(\alpha)$.
\end{theorem}
\begin{proof}
First, we bound the probability that the subgraph $H$ has more than $2k$ edges.
Let $X_e$ be an indicator variable of whether the edge $e$ is included in the returned subgraph $H$.
Recall that $E_F$ and $E_I$ denote the set of fractional edges and integral edges in $x^*$ respectively.
We assume pessimistically that all edges in $E_I$ will be included in the subgraph $H$ returned by the rounding algorithm.
Then, by Markov's inequality and Lemma~\ref{l:budget},
\[
\Pr\bigg( \sum_{e \in E} X_e > 2k \bigg)
\leq \Pr \bigg( \sum_{e \in E_F} X_e > 2k-|E_I|\bigg)
\leq \frac{\E\left[ \sum_{e \in E_F} X_e \right]}{2k-|E_I|}
\leq \frac{T\alpha \sum_{e \in E_F} x^*_e}{2k-|E_I|}
\leq \frac{T\alpha}{2},
\]
where the last inequality is by $\sum_{e \in E_F} x^*_e \leq k - |E_I|$.

Next, we bound the probability that $\Reff_H(s,t) > 4\Reff_{x^*}(s,t)$.
By Markov's inequality and Lemma~\ref{l:energy},
\[
\Pr\bigg( \Reff_H(s,t) > 4\Reff_{x^*}(s,t) \bigg)
\leq \frac{1}{4} \left( 1- \frac{1}{T} + \frac{1}{T\alpha} \right)
= \frac{T\alpha + 1}{4T\alpha} -\frac{1}{4T} 
\leq \frac{T\alpha + 1}{4T\alpha} - \Omega(\alpha),
\]
where the last inequality is because $T = \lfloor 1/\alpha \rfloor \leq 1/\alpha$.

To prove the lemma, it remains to show that
\[
\frac{T\alpha}{2} + \frac{T\alpha + 1}{4T\alpha} \leq 1
\quad \iff \quad
2(T\alpha)^2 - 3(T\alpha) + 1 = (2T\alpha - 1)(T\alpha - 1) \leq 0,
\]
which follows from Lemma~\ref{l:T}.
\end{proof}

\subsection{Constant Factor Approximation} \label{ss:constant}


We showed that the randomized path rounding algorithm is a bicriteria approximation algorithm.
To achieve a true approximation algorithm, a natural idea is to scale down the budget from $k$ to $k/2$ and apply the randomized path rounding algorithm.
The following lemma takes care of the case of $k/2 < d_{st}$, when the shortest path assumption does not hold after scaling, by showing that simply returning a shortest $s$-$t$ path is already a good enough approximation.

\begin{lemma} \label{l:path}
When the budget $k$ is at least the length $d_{st}$ of a shortest $s$-$t$ path, any $s$-$t$ shortest path is a $(k/d_{st})$-approximate solution for the $s$-$t$ effective resistance network design problem.
\end{lemma}
\begin{proof}
When $k \geq d_{st}$, a $s$-$t$ shortest path is a feasible solution to the problem with $s$-$t$ effective resistance at most $d_{st}$. 
To prove the lemma, we will show that $\Reff_x(s,t) \geq d_{st}^2/k$ for any feasible solution $x$ to \eqref{eq:P}, and so an $s$-$t$ shortest path is already a $(k/d_{st})$-approximation.

Let $G_x$ be the graph $G$ with fractional conductance $x_e$ on each edge $e \in E$.
To show a lower bound on $\Reff_x(s,t)$, we identify the vertices in $G_x$ to a form a path graph $P_x$ as follows:
For each $i \geq 0$, let $U_i$ be the set of vertices in $G$ with shortest path distance $i$ to $s$, where the shortest path distance is defined where each edge in $G$ is of length one.
First, for each $0 \leq i \leq d_{st} - 1$, we identify the vertices in $U_i$ to a single vertex $u_i$.
Then, we identify all the vertices in $\cup_{i \geq d_{st}} U_i$ to a single vertex $u_{d_{st}}$.
The path graph $P_x$ has vertex set $\{u_0, \ldots, u_{d_{st}}\}$ and edge set $\{ab \in E~|~a \in U_i {\rm~and~} b \in U_{i+1} {\rm~for~} 0 \leq i \leq d_{st}-1\}$.
For each edge $e$ in $P_x$, its conductance $x_e$ in $P_x$ is the same as that in $G_x$.
As an electrical network, identifying two vertices $uv$ is equivalent to adding an edge of resistance zero between $u$ and $v$.
So, it follows from Rayleigh's monotonicity principle (Fact~\ref{f:monotonicity}) that $\Reff_{G_x}(s,t) \geq \Reff_{P_x}(u_0,u_{d_{st}})$ as $s \in U_0$ and $t \in U_{d_{st}}$.

As $P_x$ is a series-parallel graph, we can compute $\Reff_{P_x}(s,t)$ directly.
For each $1 \leq i \leq d_{st}$,
let $E_i$ be the set of parallel edges connecting $u_{i-1}$ and $u_i$ in $P_x$,
and $c_i = \sum_{e \in E_i} x_e$ be the effective conductance between $u_{i-1}$ and $u_i$ in $P_x$.
Then, by Fact~\ref{f:resistance-SP}, 
\[
\Reff_{P_x}(u_{i-1}, u_i) = \frac{1}{c_i} 
\quad {\rm and} \quad
\Reff_{P_x}(u_0, u_{d_{st}}) = \sum_{i=1}^{d_{st}} \Reff_{P_x}(u_{i-1}, u_i) = \sum_{i=1}^{d_{st}} \frac{1}{c_i}.
\]
Note that $\sum_{i=1}^{d_{st}} c_i = \sum_{i=1}^{d_{st}} \sum_{e \in E_i} x_e \leq \sum_{e \in E} x_e \leq k$ for any feasible solution $x$.
Using Cauchy-Schwarz inequality, 
\[
d_{st} 
= \sum_{i=1}^{d_{st}} \sqrt{c_i} \cdot \frac{1}{\sqrt{c_i}} 
\leq \sqrt{\sum_{i=1}^{d_{st}} c_i} \cdot \sqrt{\sum_{i=1}^{d_{st}} \frac{1}{c_i}} 
\leq \sqrt{k} \cdot \sqrt{\Reff_{P_x}(u_0,u_{d_{st}})}.
\]
Therefore, we conclude that $\Reff_{G_x}(s,t) \geq \Reff_{P_x}(u_0,u_{d_{st}}) \geq d_{st}^2/k$.
\end{proof}

We are ready to prove our main approximation result.

\begin{theorem} \label{t:constant}
Suppose the input instance satisfies the conditions in Assumption~\ref{a:assumption}.
There is a polynomial time $8$-approximation algorithm for the $s$-$t$ effective resistance network design problem.
\end{theorem}
\begin{proof}
If the budget $k \leq 2 d_{st}$, then Lemma~\ref{l:path} shows that simply returning an $s$-$t$ shortest path would give a $2$-approximation. 
Henceforth, we assume $k \geq 2 d_{st}$.

Let ${\rm opt}(k)$ be the objective value of an optimal solution $x^*$ to the convex program \eqref{eq:P} with budget $k$, so $\Reff_{x^*}(s,t) = {\rm opt}(k)$.
As $\frac12 x^*$ is a feasible solution to \eqref{eq:P} with budget $\frac12 k$,
by Thomson's principle,
\[
{\rm opt}\left(\frac{k}{2}\right) 
\leq \Reff_{\frac12 x^*}(s,t)
= b_{st}^T \Big( \sum_{e \in E} \frac{x^*_e}{2} b_e b_e^T \Big)^\dagger b_{st} 
= 2 b_{st}^T \Big( \sum_{e \in E} x^*_e b_e b_e^T \Big)^\dagger b_{st} 
= 2 \Reff_{x^*}(s,t) 
= 2 {\rm opt}(k).
\]

Given the original budget $k \geq 2d_{st}$, our algorithm is to find an optimal solution $z^*$ to \eqref{eq:P} with budget $k/2 \geq d_{st}$, and use the path-rounding algorithm with input $z^*$ to return a subgraph $H$.
By Theorem~\ref{t:bicriteria}, with probability $\Omega(\alpha)$,
the subgraph $H$ satisfies
\[
|E(H)| \leq 2\sum_{e \in E} z^*_e \leq 2\left(\frac{k}{2}\right) = k
\quad {\rm and} \quad
\Reff_H(s,t) \leq 4 {\rm opt}\left(\frac{k}{2}\right) 
\leq 8 {\rm opt}(k),
\]
and so $H$ is an $8$-approximate solution to the $s$-$t$ effective resistance network design problem.

Finally, we consider the time complexity of the algorithm.
The number of iterations in the path rounding algorithm is $O(1/\alpha)$, 
and we need to run the path rounding algorithm $O(1/\alpha)$ times to boost the success probability to a constant.
This is a randomized polynomial time algorithm when $\alpha = \Omega(1/m)$.

In the following, we show that when $\alpha \leq 1/(4m)$, it is easy to obtain a $2$-approximate solution without running the path-rounding algorithm.
Let $x^*$ be an optimal solution to \eqref{eq:P} with budget $k$,
and $f^*$ be the unit $s$-$t$ electrical flow supported in $x^*$.
Let ${\mathcal P}$ be the flow decomposition of $f^*$ as in Lemma~\ref{l:decomposition}.
We call a path $p \in {\mathcal P}$ an integral path if every edge $e \in p$ has $x^*_e=1$; otherwise we call $p$ a fractional path.
When $\alpha \leq 1/(4m)$, we simply return the union of all integral paths as our solution $H$.
Clearly, $H$ has at most $k$ edges as it only contains integral edges. 
Next, we bound $\Reff_H(s,t)$ by the energy of the flow supported in the integral paths.
By Lemma~\ref{l:optimality}, an edge $e$ with $x^*_e < 1$ has $f^*_e = \alpha x^*_e < \alpha \leq 1/(4m)$.
This implies that each fractional path $p$ has $v_p \leq 1/(4m)$.
Since ${\mathcal P}$ has at most $m$ paths (Lemma~\ref{l:decomposition}),
the total flow in the fractional paths is at most $1/4$,
and thus the total flow in the integral paths is at least $3/4$.
By scaling the flow supported in the integral paths to a one unit $s$-$t$ flow, we see that
\[
\Reff_H(s,t) 
\leq \frac{\ene_{x^*}(f^*)}{(3/4)^2} 
\leq 2 \ene_{x^*}(f^*)
= 2 \Reff_{x^*}(s,t).
\]
To summarize, in all cases including $k < 2d_{st}$ and $\alpha \leq 1/(4m)$, there is a polynomial time algorithm to return an $8$-approximate solution to the $s$-$t$ effective resistance network design problem.
\end{proof}

We make two remarks about improvements of Theorem~\ref{t:constant}.

\begin{remark}[Approximation Ratio] \label{r:five}
The analysis of the 8-approximation algorithm is not tight. 
By a more careful analysis of the expected energy in Lemma~\ref{l:energy} and the short path idea used in the next subsection, we can show that the approximation guarantee of the same algorithm in Theorem~\ref{t:constant} is less than $5$.
However, the analysis is quite involved and not very insightful, 
so we have decided to omit those details and only keep the current analysis.
\end{remark}

\begin{remark}[Deterministic Algorithm]
Using the standard pessimistic estimator technique, 
we can derandomize the path-rounding algorithm to obtain a deterministic $8$-approximation algorithm. 
The analysis is standard and we omit the details that would take a few pages.
\end{remark}

\subsection{The Large Budget Case} \label{ss:short}

In this subsection, we show how to modify the algorithm in Theorem~\ref{t:constant} to achieve a better approximation ratio when the budget is much larger than the $s$-$t$ shortest path distance. 

The observation is that when $k \gg d_{st}$, then $\alpha$ is small by Lemma~\ref{l:alpha}, and so there are many iterations in the path-rounding algorithm.
Since each iteration is independent, we can use Chernoff-Hoeffding's bound to prove a stronger bound on the probability that the number of edges in the returned solution is significantly more than $k$ (which outperforms the bound proved in Lemma~\ref{l:budget} using Markov's inequality). 
We can then show that the expected $s$-$t$ effective resistance is close to two times the optimal value by arguments similar to the proof of Lemma~\ref{l:energy}.

\subsubsection*{Modified Rounding Algorithm}

For our analysis, we slightly modify the path-rounding algorithm to ignore ``long'' paths in the flow decomposition, so that we have a worst case bound to apply Chernoff-Hoeffding's bound.
Unlike the flow decomposition in Lemma~\ref{l:decomposition}, the short path flow decomposition definition is specific to the electrical flow of an optimal solution to \eqref{eq:P}.
In the following definition, $c$ is a parameter which will be set to be $1/\eps > 1$ to achieve a $(2+O(\eps))$-approximation.

\begin{definition}[Short Path Decomposition of Electrical Flow of Optimal Solution] \label{d:decomposition-short}
Let $x^*$ be an optimal solution to the convex program \eqref{eq:P}.
Let $f^*$ be the unit $s$-$t$ electrical flow supported in $x^*$.
Let $\alpha$ be the flow-conductance ratio defined in Lemma~\ref{l:optimality}.

Let ${\mathcal P}^*$ be a flow decomposition of $f^*$ as defined in Lemma~\ref{l:decomposition}.
Let $x^*_F := \sum_{e \in E_F} x^*_e$ be the total fractional value on the fractional edges $E_F$ in the optimal solution $x^*$.

We call a path $p \in {\mathcal P}^*$ a long path if $p$ has at least $c \alpha  x^*_F$ edges in $E_F$, i.e. $|p \cap E_F| \geq c \alpha x^*_F$.
Otherwise we call a path $p \in {\mathcal P}^*$ a short path.

Let $\mathcal{P} := \{p \in \mathcal{P}^* \mid p {\rm~is~a~short~path}\}$ be the collection of short paths in $\mathcal{P}^*$.
Let $f_{\mathcal P} := \sum_{p \in {\mathcal P}} v_p \chi_p$ be the $s$-$t$ flow defined by the short paths,
and $v_{\mathcal P} := \sum_{p \in {\mathcal P}} v_p$ be the total flow value of $f_{\mathcal P}$.
\end{definition}

The modified algorithm is very similar to the randomized path-rounding algorithm in Section~\ref{ss:bicriteria}.
The only difference is that we only sample the paths in the short path flow decomposition in Definition~\ref{d:decomposition-short},
and we adjust the sampling probability of a path $p$ to $v_p / v_{\mathcal P}$ so that the sum is one.

\begin{framed}
{\bf Randomized Short Path Rounding Algorithm}
\begin{enumerate}
\item Let $x^*$ be an optimal solution to the convex program \eqref{eq:P}.
Let $f^*$ be the unit $s$-$t$ electrical flow supported in $x^*$.
Let $\alpha$ be the flow-conductance ratio defined in Lemma~\ref{l:optimality}.

\item Compute a short path flow decomposition ${\mathcal P}$ of $f^*$ as described in Definition~\ref{d:decomposition-short}.

\item For $i$ from $1$ to $T=\lfloor 1/\alpha \rfloor$ do
  \begin{itemize}
     \item Let $P_i$ be a random path from ${\mathcal P}$ where each path $p \in {\mathcal P}$ is sampled with probability $v_p/v_{\mathcal P}$.
  \end{itemize}

\item Return the subgraph $H$ formed by the edge set $\cup_{i=1}^T P_i$.

\end{enumerate}
\end{framed}

The following simple lemma shows that the total flow on the long paths is negligible when $c$ is large, which will be useful in the analysis.

\begin{lemma} \label{l:flow-value}
For the short path flow decomposition in Definition~\ref{d:decomposition-short}, $v_{\mathcal P} \geq 1 - \frac{1}{c}$.
\end{lemma}
\begin{proof}
Using $\alpha x_e^* = f_e^*$ for $e \in E_F$ from Lemma~\ref{l:optimality} and the properties of the flow decomposition ${\mathcal P}^*$ of $f^*$ in Lemma~\ref{l:decomposition},
\[
\alpha x^*_F 
= \sum_{e \in E_F} f^*_e 
= \sum_{p \in {\mathcal P}^*} v_p \cdot |p \cap E_F| \geq \sum_{p \in {\mathcal P}^*-{\mathcal P}} v_p \cdot |p \cap E_F|
\geq c \alpha x_F^* \sum_{p \in {\mathcal P}^*-{\mathcal P}} v_p
= c \alpha x_F^* (1-v_{\mathcal P}),
\]
where the last inequality is by the definition of long paths and the last equality is because $f^*$ is a unit $s$-$t$ flow.
\end{proof}

\subsubsection*{Analysis of Approximation Guarantee}

First, we consider the expected $s$-$t$ effective resistance of the returned subgraph $H$.
For intuition, we can think of the modified rounding algorithm as applying the rounding algorithm in the scaled flow $f_{\mathcal P} / v_{\mathcal P}$, and so it should follow from Lemma~\ref{l:energy} that
\[
\E[\Reff_H(s,t)] 
\leq 2 \ene_{x^*}\left( \frac{f_{\mathcal P}}{v_{\mathcal P}} \right) 
= \frac{2}{v_{\mathcal P}^2} \ene_{x^*}(f_{\mathcal P})
\leq \frac{2}{v_{\mathcal P}^2} \ene_{x^*}(f^*)
= \frac{2}{v_{\mathcal P}^2} \Reff_{x^*}(s,t),
\]
which will be at most $(2+O(\eps)) \Reff_{x^*}(s,t)$ when $c = 1/\eps$ from Lemma~\ref{l:flow-value}.

We cannot directly apply Lemma~\ref{l:energy} as stated, as the flow $f_{\mathcal P}$ does not satisfy the flow-conductance ratio $\alpha$ in Lemma~\ref{l:optimality}, but essentially the same proof will work to get the same conclusion (but not exactly the same intermediate step).

\begin{lemma} \label{l:energy-short}
Suppose the input instance satisfies the conditions in Assumption~\ref{a:assumption}.
Let $x^*$ be an optimal solution to \eqref{eq:P} 
and $f^*$ be the unit $s$-$t$ electrical flow supported in $x^*$.
The expected $s$-$t$ effective resistance of the subgraph $H$ returned by the randomized short path rounding algorithm is 
\[
\E\left[{\rm Reff}_H(s,t)\right] 
\leq \frac{2}{v_{\mathcal P}^2} \ene_{x^*}(f^*)
= \frac{2}{v_{\mathcal P}^2} \Reff_{x^*}(s,t),
\]
where $\mathcal P$ is the short path flow decomposition of $f^*$ as described in Definition~\ref{d:decomposition-short}.
\end{lemma}

The main difference of the analysis is to apply the following Hoeffding's inequality (instead of Markov's inequality) to bound the probability that the returned subgraph has significantly more than $k$ edges.

\begin{fact}[Hoeffding's Inequality] \label{f:hoeffding}
Let $X_1, \dots, X_n\in [0, M]$ be independent random variables. Let $X = \sum_{i=1}^n X_i$, and $\mu = \E[X]$, then for any $\delta > 0$,
\[
\Pr(X  \geq (1+\delta) \mu) \leq \exp\left( -\frac{2\delta^2 \mu^2}{n M^2}\right).
\]
\end{fact}

\begin{lemma} \label{l:budget-short}
Suppose the input instance satisfies the conditions in Assumption~\ref{a:assumption}.
Let $x^*$ be an optimal solution to \eqref{eq:P} and $f^*$ be the unit $s$-$t$ electrical flow supported in $x^*$.
Let $H$ be the subgraph returned by the randomized short path rounding algorithm given $x^*$ as input, and $|E(H)|$ be the number of edges in $H$.
Then, for any $\delta > 0$,
\[\Pr\left( |E(H)| \geq (1+\delta) k \right) \leq \exp\left( - \frac{2 \delta^2}{c^2 \alpha}\right),\]
where $c$ is the parameter in the short path flow decomposition in Definition~\ref{d:decomposition-short} and $\alpha$ is the flow-conductance ratio of $f^*$ and $x^*$ as defined in Lemma~\ref{l:optimality}.
\end{lemma}
\begin{proof}
As in Lemma~\ref{l:budget}, we assume pessimistically that all integral edges $E_I$ will be included in $H$, and so we focus on the fractional edges $E_F$.
Let $X_{i,e}$ be the indicator variable of whether the edge $e$ is sampled in the $i$-th iteration of the short path rounding algorithm,
and $X_{i,F} := \sum_{e \in E_F} X_{i,e}$ be the total number of fractional edges sampled in the $i$-th iteration.
Let $X_F$ be the total number of fractional edges in $H$. Note that $X_F \leq \sum_{i=1}^T X_{i,F}$, since if some fractional edge was sampled in different iterations, we only count it once in $X_F$.
By linearity of expectation, 
$\E[X_F] \leq \sum_{i=1}^T \E[X_{i,F}]$.

Let ${\mathcal P}^*$ be the flow path decomposition of $f^*$ in Lemma~\ref{l:decomposition}, and ${\mathcal P}$ be the short path flow decomposition of $f^*$ as described in Definition~\ref{d:decomposition-short}.
For an edge $e$, recall that $(f_{\mathcal P})_e := \sum_{p \in {\mathcal P}: p \ni e} v_p$ is the total flow value on $e$ from the short paths in ${\mathcal P}$.
As we scaled the probability of each path by $1/v_{\mathcal P}$ in the rounding algorithm, the probability that edge $e$ is sampled in the $i$-th iteration is $(f_{\mathcal P})_e / v_{\mathcal P}$.
Let $\ol{f}_e := f^*_e - (f_{\mathcal P})_e$ be the total flow value on $e$ from the long paths in ${\mathcal P}^* - {\mathcal P}$.
The expected value of $X_{i,F}$ is
\[
\E[X_{i,F}]
= \sum_{e \in E_F} \E[X_{i,e}]
= \sum_{e \in E_F} \frac{(f_{\mathcal P})_e}{v_{\mathcal P}}
= \sum_{e \in E_F} \frac{f^*_e - \ol{f}_e}{v_{\mathcal P}}
= \sum_{e \in E_F} \frac{\alpha x^*_e - \ol{f}_e}{v_{\mathcal P}}
\]
By the definition of the long paths,
\[
\sum_{e \in E_F} \ol{f}_e 
= \sum_{e \in E_F} \sum_{p \in {\mathcal P}^* - {\mathcal P}} v_p
= \sum_{p \in {\mathcal P}^* - {\mathcal P}} v_p \cdot |p \cap E_F|
\geq c \alpha x^*_F \sum_{p \in {\mathcal P}^* - {\mathcal P}} v_p
= c \alpha x^*_F (1-v_{\mathcal P}),
\]
where we recall that $x^*_F = \sum_{e \in E_F} x^*_e$.
Therefore,
\[
\E[X_{i,F}]
= \sum_{e \in E_F} \frac{\alpha x^*_e - \ol{f}_e}{v_{\mathcal P}}
\leq \alpha x^*_F \cdot \frac{1 - c + cv_{\mathcal P}}{v_{\mathcal P}}
= \alpha x^*_F \cdot (c - \frac{c-1}{v_{\mathcal P}})
\leq \alpha x^*_F,
\]
where the last inequality uses that $v_{\mathcal P} \leq 1$ and $c > 1$.
It follows that $\E[X_F] \leq T \alpha x^*_F \leq x^*_F$.

As each iteration is independent, the random variables $X_{i,F}$ for $1 \leq i \leq T$ are independent.
Since we only use short paths, the maximum value of each $X_{i,F}$ is at most $c \alpha x^*_F$.
So we can apply Hoeffding's inequality to show that
\[
\Pr\left( X_F \geq (1+\delta) x^*_F \right) 
\leq \exp\left( - \frac{2\delta^2 (x^*_F)^2}{T c^2 \alpha^2 (x^*_F)^2} \right)
\leq \exp\left( - \frac{2\delta^2}{c^2 \alpha} \right).
\]
Let $X_I$ be the total number of integral edges in $H$.
As $X_I \leq |E_I|$, we conclude that
\[
\Pr\Big( |E(H)| \geq (1+\delta) k \Big)
= \Pr\Big( X_I + X_F \geq (1+\delta) (|E_I| + x^*_F)\Big)
\leq \Pr \Big( X_F \geq (1+\delta) x^*_F\Big)
\leq \exp\left( - \frac{2\delta^2}{c^2 \alpha} \right).
\]
\end{proof}

As in Section~\ref{ss:bicriteria}, we can combine Lemma~\ref{l:budget-short} and Lemma~\ref{l:energy-short} to show that the randomized short path rounding algorithm is a bicriteria approximation algorithm.

\begin{theorem} \label{t:bicriteria-short}
Suppose the input instance satisfies the conditions in Assumption~\ref{a:assumption}.
Suppose further that $k \geq d_{st}/\eps^{10}$, 
where $\eps > 0$ is an error parameter satisfying $\eps \leq \eta$ for a small constant $\eta$.
Let $x^*$ be an optimal solution to \eqref{eq:P}.
Given $x^*$, the randomized short path rounding algorithm with $c = 1/\eps$ will return a subgraph $H$ with at most $(1+\eps)k$ edges and $\Reff_H(s,t) \leq (2+10\eps) \cdot \Reff_{x^*}(s,t)$ with probability at least $\eps$.
\end{theorem}

\begin{proof}
The additional assumption $k \geq d_{st}/\eps^{10}$ implies that $\alpha \leq \eps^5$ by Lemma~\ref{l:alpha}.

Setting $c = 1/\eps$ and $\delta = \eps$,
it follows from Lemma~\ref{l:budget-short} that
\[
\Pr ( |E(H)| \geq (1+\eps)k )
\leq \exp\left( - \frac{2\delta^2}{c^2 \alpha} \right)
\leq \exp\left( - \frac{2}{\eps} \right) 
< \eps,
\]
where the last inequality holds for $\eps > 0$.

Since $c = 1/\eps$, Lemma~\ref{l:flow-value} implies that $v_{\mathcal P} \geq 1-\eps$ for the short path flow decomposition in Definition~\ref{d:decomposition-short}.
Using Markov's inequality and Lemma~\ref{l:energy-short}, for sufficiently small $\eps$ we have
\begin{eqnarray*}
\Pr \Big( \Reff_H(s,t) \geq (2+10\eps) \cdot \Reff_{x^*}(s,t) \Big)
& \leq & \frac{\E\left[ \Reff_H(s,t) \right]}{(2+10\eps) \cdot \Reff_{x^*}(s,t)}
\\
& \leq & \frac{2}{v_{\mathcal P}^2 (2+10\eps)} 
\leq \frac{2}{(1-\eps)^2 (2+10\eps)} < 1-2\eps
\end{eqnarray*}
Therefore, with probability at least $\eps$, the subgraph $H$ returned by the randomized short path rounding algorithm satisfies both properties.
\end{proof}

Using the same arguments as in Section~\ref{ss:constant}, we can turn the above bicriteria approximation algorithm into a true approximation algorithm.

\begin{theorem}
Suppose the input instance satisfies the conditions in Assumption~\ref{a:assumption}.
Suppose further that $k \geq 2d_{st}/\eps^{10}$, 
where $\eps > 0$ is an error parameter satisfying $\eps \leq \eta$ for a small constant $\eta$.
There is a polynomial time $(2+O(\eps))$-approximation algorithm for the $s$-$t$ effective resistance network design problem.
\end{theorem}
\begin{proof}
As in the proof of Theorem~\ref{t:constant}, we apply the bicriteria approximation algorithm in Theorem~\ref{t:bicriteria-short} with input $x^*$, an optimal solution to \eqref{eq:P} with the scaled-down budget $k/(1+\eps)$, to return a subgraph $H$.
As the new budget $k/(1+\eps)$ is still greater than $d_{st}/\eps^{10}$, by Theorem~\ref{t:bicriteria-short}, with probability at least $\eps$ the subgraph $H$ satisfies $|E(H)| \leq (1+\eps) k / (1+\eps) = k$ and
\[
\Reff_H(s,t) \leq \big(2+O(\eps)\big) \cdot {\rm opt}\left( \frac{k}{1+\eps} \right)
\leq \big(2+O(\eps)\big)(1+\eps) \cdot {\rm opt}(k) 
\leq \big( 2+O(\eps) \big) \cdot {\rm opt}(k),
\]
where we used the notations and arguments in Theorem~\ref{t:constant}.

For the time complexity, note that $\alpha \leq \eps^5$ by Lemma~\ref{l:alpha} and the large budget assumption, and so we can assume that $\eps^5 \geq \alpha \geq 1/(4m)$, as otherwise there is a simple $2$-approximation algorithm in the case $\alpha \leq 1/(4m)$ described in Theorem~\ref{t:constant}.
Therefore, the success probability can be boosted to a constant in polynomial number of executions of the bicriteria algorithm in Theorem~\ref{t:bicriteria-short}.
\end{proof}

\subsection{Cost Minimization with $s$-$t$ Effective Resistance Constraint} \label{ss:dual}

In this subsection, we consider a ``dual'' problem of the $s$-$t$ effective resistance minimization problem.
In the dual problem, we are given a graph $G=(V,E)$ and a target effective resistance $R$, and the objective is to find a subgraph $H$ of minimum number of edges such that $\Reff_H(s,t) \leq R$. 
The same NP-hardness proof in Section~\ref{ss:NPc} can be used to show that the dual problem is NP-complete. 

Using the same techniques for the $s$-$t$ effective resistance minimization problem, we can obtain a constant factor bicriteria approximation algorithm for this problem.
As the proofs are very similar, we will just state the results and highlight the differences.
The main difference is that the convex program has unbounded integrality gap, and as a consequence we cannot turn the bicriteria approximation algorithm into a true approximation algorithm as in the $s$-$t$ effective resistance network design problem.
Using the same technique as in Theorem~\ref{t:constant}, however, we can return an $8$-approximation to the optimal number of edges without violating the effective resistance constraint, if we are allowed to buy up to four copies of the same edge (see Theorem~\ref{t:constant-dual}).


\subsubsection*{Convex Programming Relaxation}

We consider the following natural convex programming relaxation for the dual problem.
\begin{equation}
\begin{aligned}
& \underset{x \in \R^m}{\min} 
& & \sum_{e \in E} x_e \\ 
& \text{subject to} & & \Reff_x(s,t) = b_{st}^T L_x^\dagger b_{st} \leq R, \\
&&& 0 \leq x_e \leq 1 \qquad \forall e \in E.
\end{aligned} \tag{DCP} \label{eq:DP}
\end{equation}


\subsubsection*{Integrality Gap Examples}


Unlike the $s$-$t$ effective resistance network design problem, 
the convex program \eqref{eq:DP} has unbounded integrality gap. 
Consider the following example in Figure~\ref{fig:dual_intgap}, 
where the top path has length $n\!-\!1$, and the bottom path has only one edge. 
The target effective resistance is $R = (n\!-\!1)^2/((n\!-\!1)^2 + \eps)$ for some constant $\eps > 0$. 
Since $R < 1$, to satisfy the effective resistance constraint, any integral solution must contain both paths and thus has cost $n$. 
However, the fractional solution can set $x_e = \eps / (n\!-\!1)$ for each edge in the top path and set $x_e = 1$ for the bottom edge. 
It can be checked that this fractional solution satisfies the constraint, 
and the total cost is 1+$\eps$. 
Therefore, the integrality gap of this example is $\Omega(n)$.

\begin{figure}[!ht]
\centering
\resizebox{0.65\textwidth}{!}{
\begin{tikzpicture}
	\tikzset{VertexStyle/.style = {shape = circle, fill = black, minimum size = 0.2}}
	\Vertex[Math,x=0,y=0,LabelOut,Lpos=180]{s}
	\Vertex[x=10,y=0,LabelOut]{t}	
	
    \Vertex[x=2,y=1,NoLabel]{v1}
    \Vertex[x=4,y=1,NoLabel]{v2}
    \node (dots) at (5, 1) {$\dots$};
    \Vertex[x=6,y=1,NoLabel]{v3}
    \Vertex[x=8,y=1,NoLabel]{v4}
    
    \Edge(v1)(v2)
    \Edge(v3)(v4)
	\Edge(t)(s)	
    
    \Edge(s)(v1)
    \Edge(v4)(t)
	
\end{tikzpicture}
}
\caption{$\Omega(n)$ integrality gap example.} \label{fig:dual_intgap}
\end{figure}
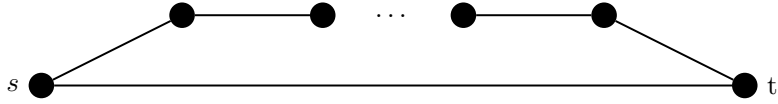

\subsubsection*{Optimal Solutions}

Although the convex program \eqref{eq:DP} has a large integrality gap, 
the same rounding technique can be used to obtain a constant factor bicriteria approximation algorithm. 
%
%
%
%
Exactly the same characterization of the optimality conditions as in the $s$-$t$ effective resistance network design problem holds, such that any optimal solution satisfies the flow-conductance ratio $\alpha > 0$ as described in Lemma~\ref{l:optimality}.

Analogous to Lemma~\ref{l:alpha}, we can prove an upper bound on $\alpha$ that
\[
\alpha^2 \leq \frac{R}{d_{st}}.
\]
Analogous to Lemma~\ref{l:path}, we can prove a lower bound on any optimal solution $x$ that
\[
{\rm opt} := \sum_{e \in E} x_e \geq \frac{d^2_{st}}{R}.
\]
We can assume that $R < d_{st}$, as otherwise a shortest $s$-$t$ path is an optimal solution, and so we can assume that $0 < \alpha < 1$.

\subsubsection*{Rounding Algorithm} 

The rounding algorithm is exactly the same as in Section~\ref{ss:bicriteria}.
The same proofs as in Lemma~\ref{l:budget} and Lemma~\ref{l:energy} will imply that, with probability $\Omega(\alpha)$, the subgraph $H$ returned by the randomized path rounding algorithm satisfies
\[
|E(H)| \leq 2 \sum_{e \in E} x^*_e 
\quad {\rm and} \quad
\Reff_H(s,t) \leq 4\Reff_{x^*}(s,t),
\]
where $x^*$ is an optimal solution to \eqref{eq:DP} and so $|E(H)| \leq 2{\rm opt}$.
The same lower bound on $\alpha = \Omega(1/m)$ as described in Theorem~\ref{t:constant} applies, and so this is a randomized polynomial time algorithm.

\subsubsection*{An Alternative Bicriteria Approximation Algorithm}

In the $s$-$t$ effective resistance network design problem, we turn a bicriteria approximation algorithm into a true approximation algorithm, by scaling down the budget $k$ by a factor of two and running the bicriteria approximation algorithm.
For the proof, we argue ${\rm opt}(k/2) \leq 2 \cdot {\rm opt}(k)$ by scaling down an optimal solution $x^*$ with budget $k$ to a solution $x^*/2$ with budget $k/2$.

In the dual problem, we can also try a similar approach, by scaling down the target effective resistance $R$ by a factor of $4$ and run the bicriteria approximation algorithm.
However, we cannot argue that ${\rm opt}(R/4) \leq 4 \cdot {\rm opt}(R)$, as an optimal solution $x^*$ with effective resistance $R$ may not be able to scale up to $4x^*$ with effective resistance $R/4$ because of the capacity constraints $0 \leq x_e \leq 1$ for $e \in E$.
This approach would work if we are allowed to violate the capacity constraint by a factor of $4$.

\begin{theorem} \label{t:constant-dual}
Given an weighted input graph $G=(V,E)$, there is a polynomial time algorithm for the dual problem which returns a multi-subgraph $H$ with $|E(H)| \leq 8{\rm opt}$ and $\Reff_H(s,t) \leq R$ where there are at most $4$ parallel copies of each edge.
\end{theorem}

\section{Dynamic Programming Algorithms for Series-Parallel Graphs} \label{s:dynamic}

In this section, we will present the dynamic programming algorithms for solving the weighted $s$-$t$ effective resistance network design problem on series-parallel graphs.
We first review the definitions of series-parallel graphs in Section~\ref{ss:SP}.
Then, we present the exact algorithm in Theorem~\ref{t:SP} when every edge has the same cost in Section~\ref{ss:SP-unit}, and the fully polynomial time approximation scheme in Theorem~\ref{t:SP} in Section~\ref{ss:SP-FPTAS}.

\subsection{Series-Parallel Graphs} \label{ss:SP}

\begin{definition} [two-terminal series-parallel graph]
A two-terminal series-parallel graph (SP graph) is a graph with two distinguished vertices (the source $s$ and the target $t$) that can be constructed recursively as follows:
    \begin{itemize}
    \item Base case: A single edge $(s, t)$
    \item Compose step: If $G_1$ and $G_2$ are two series parallel graphs with source $s_i$ and target $t_i$ ($i=1,2$), then we can combine them in two ways:
      \begin{itemize}
        \item Series-composition: We identify $t_1$ with $s_2$ as the same vertex, the source of the new graph is $s_1$ and the target is $t_2$.
        \item Parallel-composition: We identify $s_1$ with $s_2$ as the same vertex and $t_1$ with $t_2$ as the same vertex, the new source is $s_1 = s_2$ and the new target is $t_1 = t_2$.
      \end{itemize}
    \end{itemize}
\end{definition}

Given the sequence of steps of constructing a series-parallel graph $G$, we can define a tree $T$ (SP-tree) as follows.

\begin{definition} [SP-tree]~
    \begin{itemize}
    \item Leaf node: If $G$ is a single edge, then $T$ is a single node containing the edge.
    \item Recurse step: $G$ is either a series-composition (S) or a parallel-composition (P) of $G_1$ and $G_2$, then $T$ is a S-node (P-node) containing $G$, and its children are roots of the SP-trees of $G_1$ and $G_2$.
    \end{itemize}
\end{definition}

\begin{figure}[ht]
	\centering
	\includegraphics[width=0.8\textwidth]{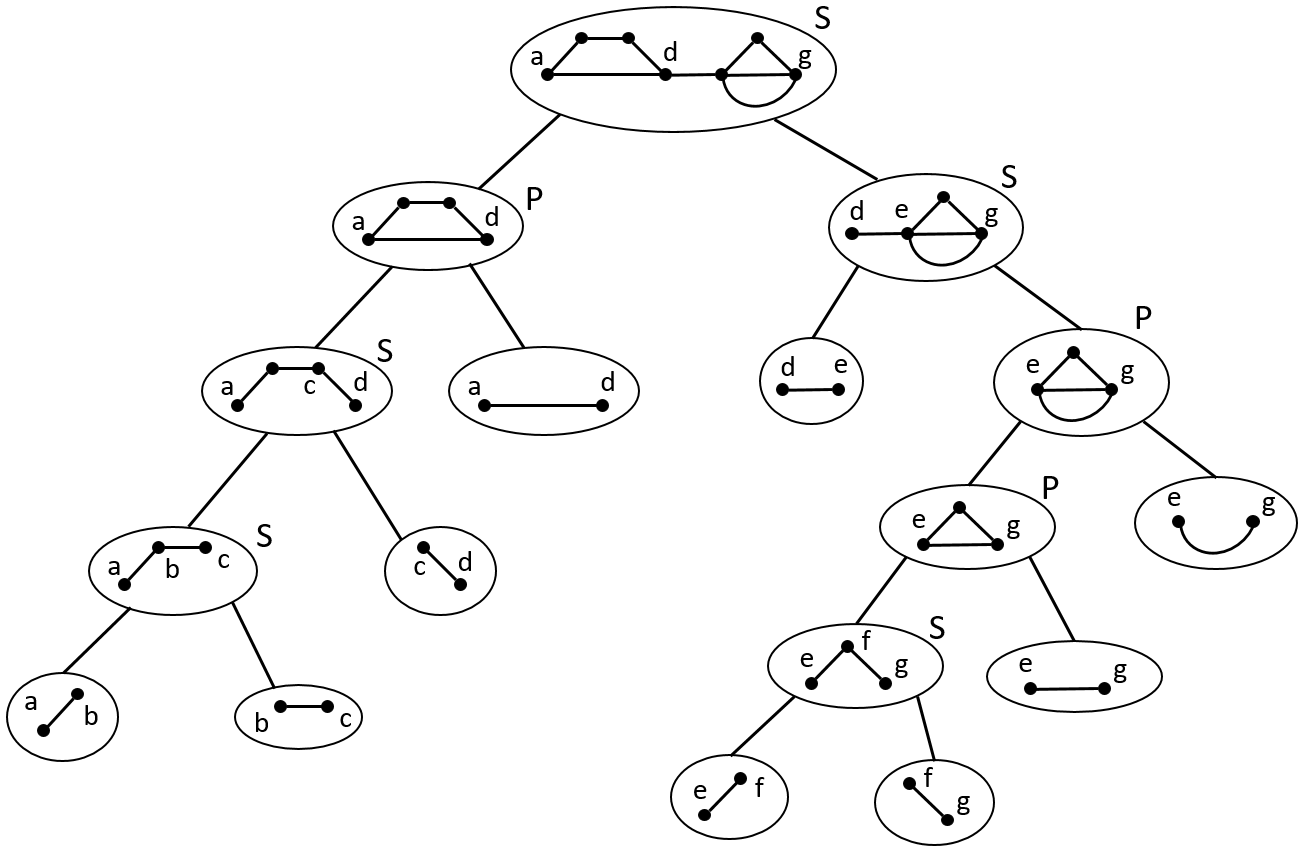}
	\caption{An example of a SP-tree.}
\end{figure}

For a tree node $v$ in a SP-tree $T$, let $G_v$ be the subgraph that $v$ represents, $s_v, t_v$ be the two terminals of $G_v$, and $v_l, v_r$ be its left and right child if $v$ is an internal node.
Note that the SP-tree is a fully binary tree with $2m\!-\!1$ nodes.

Given a two-terminal SP graph, the corresponding SP-tree can be computed in $O(n\!+\!m)$ time. 
The linear time SP-graph recognition algorithm in~\cite{valdes1979recognition} will give us the construction sequence of $G$, and we can build the SP-tree bottom-up.

\subsection{Exact Algorithm for Unit-Cost} \label{ss:SP-unit}

The following fact shows that the weighted $s$-$t$ effective resistance can be computed easily from the SP-tree.

\begin{fact} [Resistance of series-parallel network] \label{f:resistance-SP}
Let $G$ be a two-terminal SP graph and each edge $e$ has a non-negative resistance $r_e$. Let $T$ be the corresponding SP-tree. 
For every tree node $v$, we can compute the source-target effective resistance as follows.
\begin{itemize}
\item[] Leaf node: $\Reff_{G_v}(s_v, t_v) = r_e$ if $v$ is a leaf node with a single edge $e$.
\item[] S-node: $\Reff_{G_v}(s_v, t_v) = \Reff_{G_{v_l}}(s_{v_l}, t_{v_l}) + \Reff_{G_{v_r}}(s_{v_r}, t_{v_r})$.
\item[] P-node: $\Reff_{G_v}(s_v, t_v) = \dfrac{\Reff_{G_{v_l}}(s_{v_l}, t_{v_l}) \cdot \Reff_{G_{v_r}}(s_{v_r}, t_{v_r})}{\Reff_{G_{v_l}}(s_{v_l}, t_{v_l}) + \Reff_{G_{v_r}}(s_{v_r}, t_{v_r})}.$
\end{itemize}
\end{fact}

We design the dynamic programming algorithm by defining the subproblems using the SP-tree $T$.
%
For every tree node $v$ and $b = 0, 1 \dots k$, we define the subproblem
\[R(v, b) := \min_{H \subseteq  G_v} \{\Reff_H(s_v, t_v) ~|~ \sum_{e \in H} c_e \leq b \}.\]
Since we assume that every edge $e$ has cost $c_e=1$, there are at most $2mk$ subproblems, as the SP-tree has at most $2m$ nodes and there are at most $k$ possibilities for the cost of a subgraph.

It follows from the definition that $R(v_\text{root}, k)$ would be the optimal $s$-$t$ effective resistance for our problem. 	
To compute $R(v, b)$, with Fact~\ref{f:resistance-SP}, we can use the following recurrence which exhausts all possible distributions of the budget among the two children:
\[
R(v, b) = 
\begin{cases}
~~~~\infty & \text{ if $v$ is a leaf node and $b < c_e$}\\
~~~~r_e & \text{ if $v$ is a leaf node and $b \geq c_e$}\\
\displaystyle
~~~~\min_{b' = 0 \dots b} R(v_l, b') + R(v_r, b - b') & \text{ if $v$ is a S-node }\\
\displaystyle
~~~~\min_{b' = 0 \dots b} \dfrac{R(v_l, b') \cdot R(v_r, b - b')}{R(v_l, b') + R(v_r, b - b')} & \text{ if $v$ is a P-node}.
\end{cases}
\]
As there are $O(mk)$ subproblems and each subproblem can be computed in $O(k)$ time, the time complexity of this dynamic programming algorithm is $O(mk^2)$.

\subsection{Fully Polynomial Time Approximation Scheme} \label{ss:SP-FPTAS}

In this subsection, we use dynamic programming to design a fully polynomial time approximation scheme to prove Theorem~\ref{t:SP}.
In the previous subsection, we assume that every edge has the same cost to obtain an exact algorithm, by having a bounded number of subproblems in dynamic programming.
When the cost could be arbitrary, the number of subproblems can no longer be bounded by a polynomial.
Since the cost constraint must be satisfied, we do not change the cost of the edges, but instead discretize the resistance of the edges and optimize over the cost. We show that this gives an arbitrarily good approximation provided that the discretization is fine enough.

{\bf Rescaling:}
First, by rescaling we assume that $\min_e r_e=1$ and $\max_e r_e = U$ in $G$.
Let $m = |E|$ and $L=\eps/m^2$ where $\eps>0$ is the error in the approximation guarantee.
We further rescale the resistance by setting $r_e \gets r_e / L$.
This rescaling ensures that for any subgraph of $G$ in which $s$-$t$ is connected, the $s$-$t$ effective resistance is upper bounded by $Um/L$ (when all the edges are in series) and is lower bounded by $1/(mL)$ (when all the edges are in parallel).

{\bf Subproblems and Recurrence:}
Let $T$ be the SP-tree of $G$ and let $v_{\text{root}}$ be the root of $T$.
We define two similar sets of subproblems.
For every tree node $v$ and a value $R \in [1/(mL), Um/L]$,
we define the subproblem
\[C(v, R) := \min_{H \subseteq G_v} \left\{ \sum_{e \in H} c_e \mid \text{Reff}_H(s_v, t_v) \leq R \right\}.\]
Similar to the reasoning in the previous subsection, 
the subproblems satisfy the following recurrence relation:
\begin{align*}
&C(v, R) =
&\begin{cases}
c_e \text{~~~~if $v$ is a leaf node with a single edge $e$ and $R \geq r_e$} & \\
\infty \text{~~~~if $v$ is a leaf node with a single edge $e$ and $R < r_e$} &\\
\displaystyle
  \min_{R_1, R_2 \in [1/(mL), Um/L]} \{C(v_l, R_1) + C(v_r, R_2) \mid R_1 + R_2 \leq R\} & \text{if $v$ is a S-node}\\
\displaystyle
  \min_{R_1, R_2 \in [1/(mL), Um/L]} \{C(v_l, R_1) + C(v_r, R_2) \mid \dfrac{R_1 R_2}{R_1 + R_2} \leq R\} & \text{if $v$ is a P-node.}\\
\end{cases}
\end{align*}

{\bf Discretized subproblems:}
We cannot use dynamic programming to solve the above recurrence relation efficiently as there are unbounded number of subproblems.
Instead, we use dynamic programming to compute the solution of all the ``discretized'' subproblems using the same recurrence relation.
For every integer $R$ from $\lceil 1/(mL) \rceil$ to $\lceil Um/L \rceil$, we define
\begin{align*}
&\ol{C}(v, R) :=
&\begin{cases}
  c_e \text{~~~~if $v$ is a leaf node with a single edge $e$ and $R \geq \lceil r_e \rceil$} & \\
  \infty \text{~~~~if $v$ is a leaf node with a single edge $e$ and $R < \lceil r_e \rceil$} &\\
  \displaystyle
  \min_{R_1, R_2 \in \{\lceil 1/(mL) \rceil \dots \lceil Um/L \rceil\}}\{\ol{C}(v_l, R_1) + \ol{C}(v_r, R_2) \mid R_1 + R_2 \leq R\} & \text{if $v$ is a S-node}\\
  \displaystyle
  \min_{R_1, R_2 \in \{\lceil 1/(mL) \rceil \dots \lceil Um/L \rceil\}} \{\ol{C}(v_l, R_1) + \ol{C}(v_r, R_2) \mid \left\lceil \dfrac{R_1 R_2}{R_1 + R_2} \right\rceil \leq R\} & \text{if $v$ is a P-node.}\\
\end{cases}
\end{align*}
We can think of $\ol{C}(v, R)$ as the minimum cost required to select a subset of edges such that the effective resistance between $s_v$ and $t_v$ is at most $R$, when the effective resistance is rounded up to an integer during each step of the computation in the recurrence relation. 

{\bf Algorithm and Complexity:}
After computing all $\ol{C}(v, R)$, the algorithm will return 
\[\min \{ R \mid \ol{C}(v_{\text{root}}, R) \leq k \}\] 
as the approximate minimum $s$-$t$ effective resistance.
Given a tree node $v$, by trying all possible integral values of $R_1$ and $R_2$, we can compute the values of $\ol{C}(v, R)$ for each possible $R$ in $O((Um/L)^2)$ time. 
Therefore, the total running time of computing all $\ol{C}(v, R)$ is $O(m) \cdot O((Um/L)^2) = O(m^7U^2/\varepsilon^2)$.
To output the optimal edge set, we can store the optimal values of $R_1, R_2$ for each pair of $(v, R)$ to reconstruct the edge set.
    
{\bf Correctness and Approximation Guarantee:}
Since we have not changed the edge cost, the solution returned by the algorithm will have total cost at most $k$.
It remains to show that the $s$-$t$ effective resistance is at most $(1+\eps)$ times the optimal $s$-$t$ effective resistance.
For every tree node $v$ and every $b \in [0, k]$, we define
  \begin{align*}
    R(v, b) &:= \min \{ R \mid C(v, R) \leq b, R \in [1/(mL), Um/L] \} \\
    \ol{R}(v, b) &:= \min \{ R \mid \ol{C}(v, R) \leq b, R \in \{\lceil 1/(mL) \rceil, \dots \lceil Um/L \rceil\}\}.
  \end{align*}
It follows from the definitions that the optimal $s$-$t$ effective resistance is $R(v_{\text{root}}, k)$, and the output of our algorithm will be $\ol{R}(v_{\text{root}}, k)$.
The following lemma establishes the approximation guarantee.

\begin{lemma}
For every tree node $v$ and for every $b \in [0,k]$, it holds that
\[\ol{R}(v, b) \leq \left(1 + \frac{\varepsilon|E(G_v)|}{m} \right) R(v, b).\]
\end{lemma}
\begin{proof}
We prove the lemma by induction on the tree node of the SP-tree.	

{\bf Base Case:} 
Suppose $v$ is a leaf node of $T$ and $G_v$ is a graph of a single edge $e$. 
\begin{itemize}
\item For $b < c_e$, we have $\ol{R}(v, b) = R(v, b) = \infty$.
\item For $b \geq c_e$, we have $R(v,b)=r_e$ and
  \begin{align*}
    \ol{R}(v, b) & 
    = \lceil r_e \rceil \leq r_e + 1 
    = r_e + \left(\frac{\varepsilon}{m}\right) \left(\frac{1}{mL}\right) 
    \leq r_e + \frac{\varepsilon}{m}r_e 
    = \left(1 + \frac{\varepsilon|E(G_v)|}{m} \right)R(v, b),
  \end{align*}
where the second inequality uses the fact that every resistance is at least $1/(mL)$, and the last equality uses $|E(G_v)|=1$ and $r_e=R(v,b)$.
\end{itemize}

{\bf S-node:} 
Suppose $v$ is a S-node.
For every $b \in [0, k]$, we have
\begin{align*}
    \ol{R}(v, b) &= \min_{b_1, b_2 \mid b_1 + b_2 = b} \{ \ol{R}(v_l, b_1) + \ol{R}(v_r, b_2)\} \\
    &\leq \min_{b_1, b_2 \mid b_1 + b_2 = b} \left\{ \left(1 + \frac{\varepsilon|E(G_{v_l})|}{m} \right) R(v_l, b_1) +\left(1 + \frac{\varepsilon|E(G_{v_r})|}{m} \right) R(v_r, b_2)\right\}\\
    &\leq \min_{b_1, b_2 \mid b_1 + b_2 = b} \left\{ \left(1 + \frac{\varepsilon|E(G_v)|}{m} \right) (R(v_l, b_1) + R(v_r, b_2))\right\}\\
    &= \left(1 + \frac{\varepsilon|E(G_v)|}{m} \right) \min_{b_1, b_2 \mid b_1 + b_2 = b} \{(R(v_l, b_1) + R(v_r, b_2))\}\\
    &= \left(1 + \frac{\varepsilon|E(G_v)|}{m} \right) R(v, b),
  \end{align*}
where the first inequality follows from the induction hypothesis, and the second inequality follows from the fact that $\max(|E(G_{v_l})|,  |E(G_{v_r})|) \leq |E(G_v)| - 1$.

{\bf P-node:}
Suppose $v$ is a P-node.
For every $b \in [0, B]$, we have
  \begin{align*}
    \ol{R}(v, b) &= \min_{b_1, b_2 \mid b_1 + b_2 = b} \left\{ \left\lceil \dfrac{1}{1/\ol{R}(v_l, b_1) + 1/\ol{R}(v_r, b_2)} \right\rceil \right\} \\
    &\leq \min_{b_1, b_2 \mid b_1 + b_2 = b} \left\{ \left\lceil \left(1 + \frac{\varepsilon(|E(G_v)|-1)}{m} \right) \dfrac{1}{1/R(v_l, b_1) + 1/R(v_r, b_2)} \right\rceil \right\}\\
    &= \left\lceil \left(1 + \frac{\varepsilon(|E(G_v)|-1)}{m} \right) R(v, b) \right\rceil \\
    &\leq  \left(1 + \frac{\varepsilon(|E(G_v)|-1)}{m} \right) R(v, b)  + 1\\
    &= \left(1 + \frac{\varepsilon(|E(G_v)|-1)}{m} \right) R(v, b) + \frac{\varepsilon}{m}\frac{1}{mL}\\
    &\leq \left(1 + \frac{\varepsilon(|E(G_v)|-1)}{m} \right) R(v, b) + \frac{\varepsilon}{m} R(v, b)\\
    &= \left(1 + \frac{\varepsilon|E(G_v)|}{m} \right) R(v, b),
   \end{align*}
where the first inequality follows from the induction hypothesis and the fact that $\max(|E(G_{v_l})|,  |E(G_{v_r})|) \leq |E(G_v)| - 1$, 
and the last inequality holds as the minimum resistance of any subgraph is at least $1/(mL)$.

Therefore, the lemma follows by induction on the SP-tree.
\end{proof}

By substituting $v = v_{\text{root}}$ and $b=k$,
we have
\[\ol{R}(v_{\text{root}}, k) \leq \left(1 + \frac{m\varepsilon}{m} \right) R(v_{\text{root}}, k) = (1 + \varepsilon) R(v_{\text{root}}, k),\]
which completes the proof of Theorem~\ref{t:SP}.

\section{Hardness Results} \label{s:hardness}

In this section, we first prove that the $s$-$t$ effective resistance network design problem is NP-hard in Section~\ref{ss:NPc}.
Then, we prove that the weighted problem is APX-hard assuming the small-set expansion conjecture in Section~\ref{ss:SSE}.

\subsection{NP-Hardness} \label{ss:NPc}

We will prove Theorem~\ref{t:NPc} in this subsection.
The following is the decision version of the problem.
\begin{problem}[$s$-$t$ effective resistance network design] \label{p:unit}~
\begin{enumerate}
\item[] 
{\bf Input:} An undirected graph $G=(V,E)$, two vertices $s,t \in V$, and two parameters $k$ and $R$.
\item[] 
{\bf Question:} Does there exist a subgraph $H$ of $G$ with at most $k$ edges and $\Reff_H(s,t) \leq R$?
\end{enumerate}
\end{problem}

We will show that this problem is NP-complete by a reduction from the 3-Dimensional Matching (3DM) problem.

\begin{problem}[3-Dimensional Matching]~
\begin{enumerate}
\item[]
{\bf Input:} Three disjoint sets of elements $X = \{x_1, \dots, x_q\},Y = \{y_1, \dots, y_q\}, Z=\{z_1, \dots, z_q\}$; 
a set of triples ${\mathcal T} \subseteq X \times Y \times Z$ where each triple contains exactly one element in $X,Y,Z$.
\item[] 
{\bf Question:} Does there exist a subset of $q$ pairwise disjoint triples in ${\mathcal T}$?
\end{enumerate}
\end{problem}

{\bf Reduction:}
Given an instance of 3DM with $\{(X,Y,Z),{\mathcal T}\}$, 
let $\tau=|{\mathcal T}|$ and denote the triples by ${\mathcal T} = \{T_1, \ldots, T_\tau\}$.

\begin{figure}[!ht]
	\centering
	\includegraphics[width=0.7\textwidth]{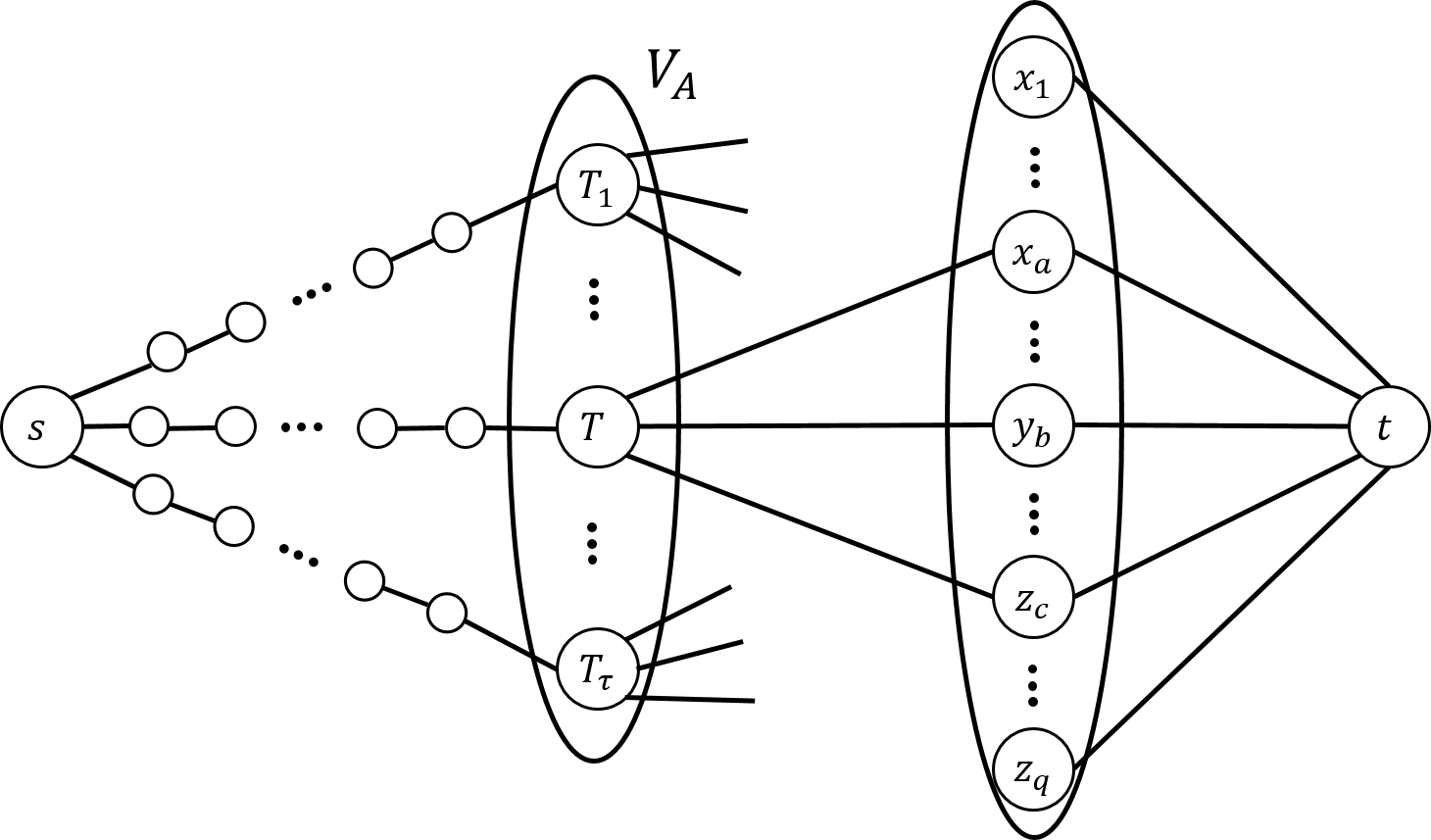}
\caption{An illustration of the construction of the graph $G$ from a 3DM instance.} \label{fig:3DM}

\end{figure}

We construct a graph $G = (V,E)$ as follows: 
\begin{enumerate}
\item[]
{\bf Vertex Set:} The vertex set $V$ is the disjoint union of five sets $\{s\}, \{t\}, V_A$, $V_B$, and $D$.
Each vertex in $V_A$ corresponds to a triple in ${\mathcal T}$, that is $V_A = \{T_1, \ldots, T_{\tau}\}$.
Each vertex in $V_B$ corresponds to an element in $X \cup Y \cup Z$, that is $V_B = \{x_1, \ldots, x_q, y_1, \ldots, y_q, z_1, \ldots, z_q\}$.
Let $l = 3\tau + 3q$.
The set $D$ consists of $\tau \cdot l$ ``dummy'' vertices $\{d_{i,j}~|~1 \leq i \leq \tau, 1 \leq j \leq l\}$. 
So, there are totally $\tau + 3q + 2 + \tau (3 \tau + 3q)$ vertices in $G$, which is polynomial in the input size of the 3DM instance.
\item[]
{\bf Edge Set:} The edge set $E$ is the disjoint union of three edge sets $F_1$, $F_2$ and $P$.
There are $3\tau$ edges in $F_1$, where we have three edges $(T,x_a)$, $(T,y_b)$ and $(T,z_c)$ for each triple $T=(x_a,y_b,z_c) \in {\mathcal T}$.
There are $3q$ edges in $F_2$, where there is an edge from each vertex in $V_B$ to $t$.
There are $\tau(l+1)$ edges in $P$, where there is a path $P_i := (s,d_{i,1},d_{i,2},\ldots,d_{i,l},T_i)$ for each triple $T_i \in {\mathcal T}$, $1 \leq i \leq \tau$.
So, there are totally $3\tau + 3q + \tau(3\tau+3q+1)$ edges in $E$, which is polynomial in the input size of the 3DM instance.
\end{enumerate}


The following claim completes the proof of Theorem~\ref{t:NPc}.

\begin{lemma}
Let $k=q(l+1)+3\tau+3q$ and $R=(3(l+1)+2)/3q$.
The 3DM instance has $q$ disjoint triples if and only if the graph $G$ has a subgraph $H$ with at most $k$ edges and $\Reff_H(s,t) \leq R$.
\end{lemma}
\begin{proof}
%
%

One direction is easy.
If there are $q$ disjoint triples in the 3DM instance, say $\{T_1, \ldots, T_q\}$, then $H$ will consist of the $q$ paths $P_1, \ldots, P_q$, the $3q$ edges in $F_1$ incident on $T_1, \ldots, T_q$, and all the $3q$ edges in $F_2$.
There are $(l+1)q + 3q + 3q \leq k$ edges in $H$,
and $\Reff_H(s,t) = (l+1)/q + 1/3q + 1/3q = (3(l+1)+2)/3q = R$, 
as in the graph in Figure~\ref{f:reduced2}.

The other direction is more interesting.
If there do not exist $q$ disjoint triples in the 3DM instance, then we need to argue that $\Reff_H(s,t) > R$ for any $H$ with at most $k$ edges.
First, note that $k < (q+1)(l+1)$, and so the budget is not enough for us to buy more than $q$ paths.
As it is useless to buy only a proper subset of a path, we can thus assume that $H$ consists of $q$ paths and all the edges in $F_1, F_2$. $H$ has a total of exactly $q(l+1) + 3\tau + 3q = k$ edges.
For any such $H$, we will argue that $\Reff_H(s,t) > R$.
Without loss of generality, assume that $H$ consists of $P_1, \ldots, P_q$ and all edges in $F_1$ and $F_2$.
As $T_1, \ldots, T_q$ are not disjoint, there are some vertices in $V_B$ that are not neighbors of $T_1 \cup \ldots \cup T_q$. Call those vertices $U$.

We consider the following modifications of $H$ to obtain $H'$, and use $\Reff_{H'}(s,t)$ to lower bound $\Reff_H(s,t)$.
For every pair of vertices in $V_B$, we add an edge of zero resistance.
For each edge incident on $T_{q+1}, \ldots, T_{\tau}$, we decrease its resistance to zero.
By the monotonicity principle, the modifications will not increase the $s$-$t$ effective resistance, as we either add edges with zero resistance or decrease the resistance of existing edges.
The modifications are equivalent to contracting the vertices with zero resistance edges in between, and so $H'$ is equivalent to the graph in Figure~\ref{f:reduced2}.
Therefore, we have $\Reff_H(s,t) \geq \Reff_{H'}(s,t) \geq R$.

\begin{figure}[!ht]
\centering
	\centering
	\includegraphics[width=0.7\textwidth]{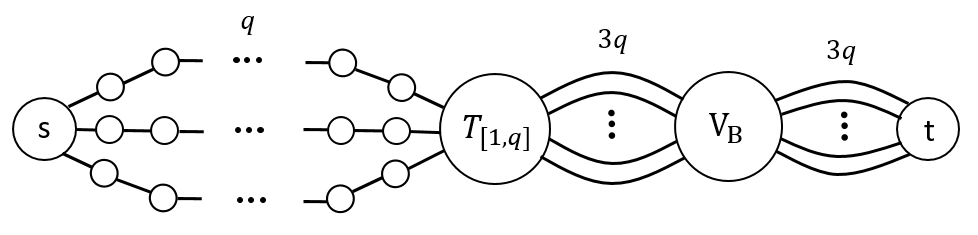}
\caption{The subgraph $H$ when the 3DM instance has $q$ disjoint triples.} \label{f:reduced2}
\end{figure} 

\begin{figure}[!ht]
\centering
	\centering
	\includegraphics[width=0.7\textwidth]{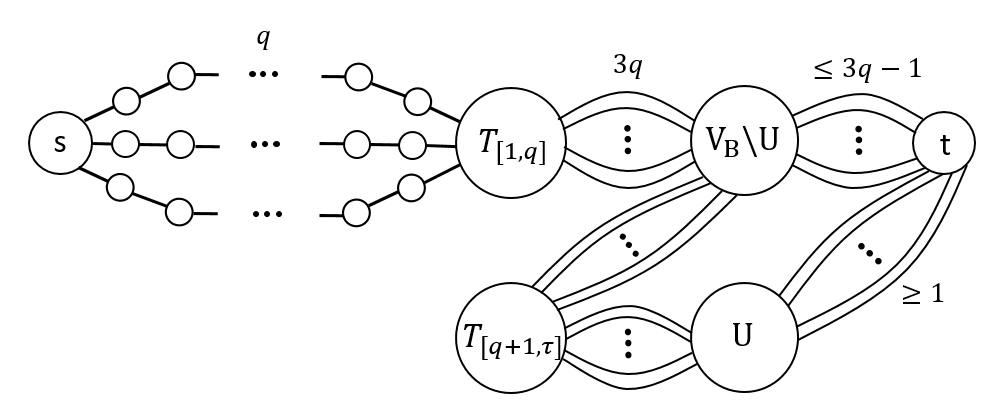}
\caption{The subgraph $H$ when $U$ is non-empty.} \label{f:reduced1}
\end{figure} 

We will prove that one of the inequalities in $\Reff_H(s,t) \geq \Reff_{H'}(s,t)\geq R$ must be strict when $U \neq \emptyset$ (Figure~\ref{f:reduced1}).
To argue the strict inequality, we look at the unit $s$-$t$ electrical flow $f$ in $H$ and consider two cases.
\begin{itemize}
\item If there exists some vertex $u \in U$ with no incoming electrical flow, then we can delete such a vertex without changing $\Reff_{H}(s,t)$.
But then in the modified graph $H'$, the number of parallel edges to $t$ is now strictly smaller than $3q$, and therefore $\Reff_{H'}(s,t) > R$.
\item If there exists some vertex $u \in U$ with some incoming electrical flow, then $f(T_j u) > 0$ for some $j \geq q+1$.
Since we have decreased the resistance of such an edge $T_j u$ to 0, the energy of $f$ in $H'$ is strictly smaller than the energy of $f$ in $H$.
By Thomson's principle, we have $\Reff_{H'}(s,t) \leq \ene_{H'}(f) < \ene_H(f) = \Reff_H(s,t)$.
\end{itemize}
Since the 3DM instance has no $q$ disjoint triples, it follows that $U \neq \emptyset$ and thus one of the above two cases must apply.
In either case, we have $\Reff_{H}(s,t) > R$ and this completes the proof of the other direction.
\end{proof}

\subsection{Improved Hardness Assuming Small-Set Expansion Conjecture} \label{ss:SSE}

In this subsection, we will prove Theorem~\ref{t:SSE} that it is NP-hard to approximate the weighted $s$-$t$ effective resistance network design problem within a factor smaller than $2$.
First, we will state the small-set expansion conjecture and its variant on bipartite graphs, and present an overview of the proof in Section~\ref{ss:conjecture}.
Next, we will reduce the bipartite small-set expansion problem to the weighted $s$-$t$ effective resistance network design problem in Section~\ref{ss:SSE-Reff},
and then reduce the small-set expansion problem to the bipartite small-set expansion problem in Section~\ref{ss:bipartite} to complete the proof.

\subsubsection{The Small-Set Expansion Conjecture and Proof Overview}
\label{ss:conjecture}

The gap small-set expansion problem is formulated by Raghavendra and Steurer~\cite{RS10}.
We use the version stated in~\cite{RST12}.

\begin{definition}[Gap Small-Set Expansion Problem~\cite{RS10,RST12}]
Given an undirected graph $G=(V,E)$, two parameters $0 < \beta < \alpha <1$ and $\delta>0$, the $(\alpha,\beta)$-gap $\delta$-small-set expansion problem, denoted by $\text{SSE}_{\delta}(\alpha, \beta)$, is to distinguish between the following two cases.
\begin{itemize}
\item {\sc Yes}: There exists a subset $S \subseteq V$ with $\vol(S) = \delta \vol(V)$ and $\phi(S) \leq \beta$.
\item {\sc No}: Every subset $S \subseteq V$ with $\vol(S) = \delta \vol(V)$ has $\phi(S) \geq \alpha$.
\end{itemize}
%
\end{definition}

It is conjectured in~\cite{RS10} that the gap small-set expansion problem becomes harder when $\delta$ becomes smaller.

\begin{conjecture}[Small-Set Expansion Conjecture \cite{RS10,RST12}]
For any $\eps \in (0,\frac12)$, there exists sufficiently small $\delta>0$ such that $SSE_\delta(1-\eps, \eps)$ is NP-hard even for regular graphs.
\end{conjecture}

It is known that the small-set expansion conjecture implies the Unique Game conjecture~\cite{RS10} and is equivalent to some variant of the Unique Game Conjecture~\cite{RST12}.

We will show the SSE-hardness of the weighted $s$-$t$ effective resistance network design problem in two steps, and use the small-set expansion problem on regular {\em bipartite} graphs as an intermediate problem.

\begin{proposition} \label{p:bipartite}
For any $\eps > 0$, there is a polynomial time reduction from $\text{SSE}_{\delta}(1-\eps,\eps)$ on $d$-regular graphs to $\text{SSE}_{\delta}(1-16\eps,\eps)$ on $d$-regular bipartite graphs.
\end{proposition}

\begin{proposition} \label{p:SSE-Reff}
Given an instance of $\text{SSE}_{\delta}(\alpha,\beta)$ on a $d$-regular bipartite graph $B$, there is a polynomial time algorithm to construct an instance of the weighted $s$-$t$ effective resistance network design problem with graph $G$ and cost budget $k$ satisfying the following properties.
\begin{itemize}
\item If $B$ is a {\sc Yes}-instance, then there is a subgraph $H$ of $G$ with cost at most $k$ and 
\[\Reff_H(s,t) \leq \frac{2}{(1-\beta)dk}.\]
\item if $B$ is a {\sc No}-instance, then every subgraph $H$ of $G$ with cost at most $k$ has 
\[\Reff_H(s,t) \geq \frac{2}{(1-\frac{\alpha}{2})dk}.\]
\end{itemize}
\end{proposition}



Theorem~\ref{t:SSE} will follow immediately from the two propositions.

\begin{theorem} \label{t:APX}
For any $\eps'>0$, it is NP-hard to approximate the weighted $s$-$t$ effective resistance network design problem to within a factor of $2 - \eps'$, assuming that $\text{SSE}_\delta(1-\eps,\eps)$ is NP-hard on regular graphs for sufficiently small $\eps>0$. 
\end{theorem}
\begin{proof}
First, given a $d$-regular instance of $\text{SSE}_\delta(1-\eps,\eps)$, we apply Proposition~\ref{p:bipartite} to obtain a $d$-regular bipartite instance of $\text{SSE}_\delta(1-16\eps,\eps)$.
Then, we apply Proposition~\ref{p:SSE-Reff} with $\alpha = 1-16\eps$ and $\beta=\eps$ and see that the ratio between the $s$-$t$ effective resistance of the {\sc No}-case and the {\sc Yes}-case is at least
\[
\frac{(1-\beta)dk}{(1-\frac{\alpha}{2})dk} = \frac{1-\eps}{\frac{1}{2} + 8\eps} = \frac{2(1-\eps)}{1+16\eps} > 2 - \eps',
\]
for sufficiently small $\eps$.
\end{proof}

We will prove Proposition~\ref{p:SSE-Reff} in Section~\ref{ss:SSE-Reff} and Proposition~\ref{p:bipartite} in Section~\ref{ss:bipartite}.

\subsubsection{From Bipartite Small-Set Expansion to weighted $s$-$t$ Effective Resistance Network Design} \label{ss:SSE-Reff}

We prove Proposition~\ref{p:SSE-Reff} in this subsection.
In the {\sc Yes}-case of bipartite SSE, we use the small dense subgraph (from the small low conductance set) to construct a small subgraph with small $s$-$t$ effective resistance.
In the {\sc No}-case of bipartite SSE, we argue that every small subgraph has considerably larger $s$-$t$ effective resistance.

\begin{figure}[!ht]
\centering
\resizebox{0.6\textwidth}{!}{
\begin{tikzpicture}
	\tikzset{VertexStyle/.style = {shape = circle, fill = black, minimum size = 0.2}}
	\Vertex[Math,x=0,y=0,LabelOut,Lpos=180]{s}
	\Vertex[Math,x=3, y=2.5,LabelOut,Lpos=90]{x_1}		
	\node (dots) at (3, 1.5) {$\vdots$};	
	\Vertex[Math,x=3, y=0,LabelOut,Lpos=145]{x_i}	
	\node (dots) at (3, -1.3) {$\vdots$};
	\Vertex[Math,x=3, y=-2.5,LabelOut,Lpos=270]{x_n}		
	\Vertex[Math, x=6, y=2.5,LabelOut,Lpos=90]{y_1}
	\node (dots) at (6, 1.5) {$\vdots$};	
	\Vertex[Math,x = 6, y = 0,LabelOut,Lpos=45]{y_j}
	\node (dots) at (6, -1.3) {$\vdots$};	
	\Vertex[Math,x=6,y=-2.5,LabelOut,Lpos=270]{y_n}
	\Vertex[x=9,y=0,LabelOut]{t}

	\Edge(y_1)(t)
	\Edge(t)(y_n)	
	
	\tikzset{EdgeStyle/.style = {}}
	\Edge(s)(x_i)
	\Edge(t)(y_j)
	
	\Edge(s)(x_1)
	\Edge(s)(x_n)

	\draw (x_1) -- (4,2.3);	
	\draw (x_1) -- (4,2.1);	
	\draw (x_1) -- (4,1.9);	

	\draw (x_n) -- (4,-2.3);	
	\draw (x_n) -- (4,-2.1);	
	\draw (x_n) -- (4,-1.9);
	
	\draw (y_1) -- (5,2.2);	
	\draw (y_1) -- (5,2);	
	\draw (y_1) -- (5,1.7);	

	\draw (y_n) -- (5,-2.5);	
	\draw (y_n) -- (5,-2);	
	\draw (y_n) -- (5,-1.7);
	
	\draw (x_i) -- (4,0.6);	
	\draw (x_i) -- (4,0);	
	\draw (x_i) -- (4,-0.6);
	
	\draw (y_j) -- (5,0.6);	
	\draw (y_j) -- (5,0);	
	\draw (y_j) -- (5,-0.6);	
	


	\draw (3,0) ellipse [x radius = 0.8, y radius = 3.3];
	\node (VX) at (3, -3.6) {$V_X$};
	\draw (6,0) ellipse [x radius = 0.8, y radius = 3.3];
	\node (VY) at (6, -3.6) {$V_Y$};

\end{tikzpicture}
}

\caption{
Reduction from bipartite small set expansion to weighted $s$-$t$ effective resistance network design. 
} \label{f:bipartite}

\end{figure}
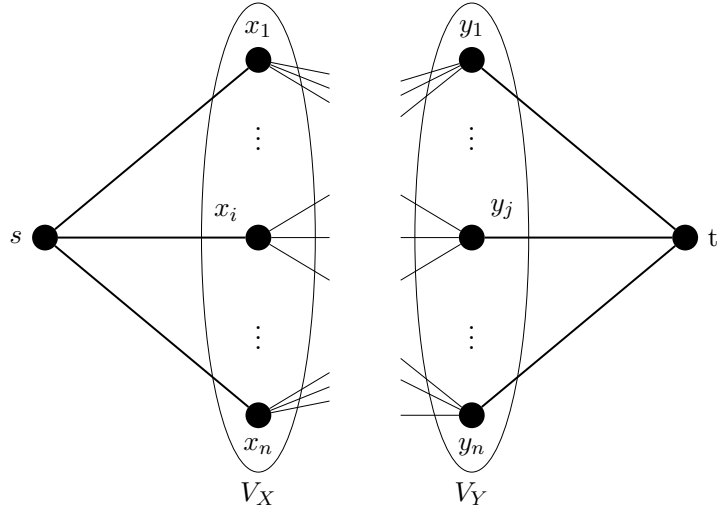

{\bf Construction:}
Given an $\text{SSE}_\delta(\alpha,\beta)$ instance with a $d$-regular bipartite graph $B=(V_X,V_Y;E_B)$, we construct an instance of the weighted $s$-$t$ effective resistance network design problem with graph $G=(V,E)$ as follows.
See Figure~\ref{f:bipartite} for an illustration.
\begin{itemize}
\item[] {\bf Vertex Set:} The vertex set $V$ of $G$ is simply the disjoint union of $\{s\}, V_X, V_Y, \{t\}$.
\item[] {\bf Edge Set:} The edge set $E$ of $G$ is the disjoint union of three edge sets $E_s, E_B, E_t$.
The edge set $E_s$ has $|V_X|$ edges, where there is an edge from $s$ to each vertex $v \in V_X$.
The edge set $E_t$ has $|V_Y|$ edges, where there is an edge from each vertex $v \in V_Y$ to $t$.
\item[] {\bf Costs and Resistances:} Every edge $e$ in $E_B$ has $c_e=0$ and $r_e=1$.  Every edge $e \in E_s \cup E_t$ has $c_e=1$ and $r_e=0$.
\item[] {\bf Budget:} The cost budget $k$ is $\delta |V_X \cup V_Y|$.
\end{itemize}

\begin{figure}[!ht]
\centering
\resizebox{0.65\textwidth}{!}{
\begin{tikzpicture}
	\def \rad{.7}
    \def \phi{30}
    \def \sx{0}
    \def \sy{0}
    \def \tx{11}
    \def \ty{0}
    \def \Xx{3}
    \def \Xy{0}
    \def \Yx{8}
    \def \Yy{0}
    \def \Zx{5.5}
    \def \Zy{2.5}
	
	\node[draw, radius=\radst, circle, fill, label=left:$s$] (s) at (\sx,\sy) {};
    \node[draw, radius=\radst, circle, fill, label=right:$t$] (t) at (\tx,\ty) {};
    \node (X) at (\Xx,\Xy) {$X$};
    \draw (\Xx,\Xy) circle [radius=\rad];
    \node (Y) at (\Yx,\Yy) {$Y$};
    \draw (\Yx,\Yy) circle [radius=\rad];
    \node (Z) at (\Zx,\Zy) {$Z$};
    \draw (\Zx,\Zy) circle [radius=\rad];
    
    \draw ($(\sx,\sy)$) -- node[above] {$|X|$ edges} ($(\Xx,\Xy) + ( 180-\phi: \rad )$);
    \node (dots) at ($(\sx/2,\sy/2) + (\Xx/2,\Xy/2) + (.4,0.1)$) {$\vdots$};
    \draw ($(\sx,\sy)$) -- ($(\Xx,\Xy) + ( -180+\phi: \rad )$);
    
    \draw ($(\Yx,\Yy) + ( \phi: \rad )$) -- node[above] {$|Y|$ edges} ($(\tx,\ty)$);
    \node (dots) at ($(\Yx/2,\Yy/2) + (\tx/2,\ty/2) + (-.4,0.1)$) {$\vdots$};
    \draw ($(\Yx,\Yy) + ( -\phi: \rad )$) -- ($(\tx,\ty)$);
    
    \draw ($(\Xx,\Xy) + ( \phi: \rad )$) -- ($(\Yx,\Yy) + ( 180-\phi: \rad )$);
    \node (dots) at ($(\Xx/2,\Xy/2) + (\Yx/2,\Yy/2) + (0,0.1)$) {$\vdots$};
    \draw ($(\Xx,\Xy) + ( -\phi: \rad )$) -- node[below] {$\geq \frac12 (1-\beta) d k$ edges} ($(\Yx,\Yy) + ( -180+\phi: \rad )$);
    
    \draw [dashed] ($(\Xx,\Xy) + ( \phi: \rad )$) -- ($(\Zx,\Zy) + ( -90-\phi: \rad )$);
    \draw [dashed] ($(\Xx,\Xy) + ( 10+2*\phi: \rad )$) -- ($(\Zx,\Zy) + ( -100-2*\phi: \rad )$);
    \node [rotate=-135] (dots) at ($(\Xx/2,\Xy/2) + (\Zx/2,\Zy/2)-(0,0.05)$) {$\vdots$};

    \draw [dashed] ($(\Yx,\Yy) + ( 180-\phi: \rad )$) -- ($(\Zx,\Zy) + ( -90+\phi: \rad )$);
    \draw [dashed] ($(\Yx,\Yy) + ( 80+\phi: \rad )$) -- ($(\Zx,\Zy) + ( -\phi+10: \rad )$);
    \node [rotate=135] (dots) at ($(\Yx/2,\Yy/2) + (\Zx/2,\Zy/2)-(0,0.05)$) {$\vdots$};
\end{tikzpicture}
}
\caption{In the {\sc Yes}-case, the solid edges are included in $H$ and the dashed edges are deleted.} \label{fig:YESReff}
\end{figure}
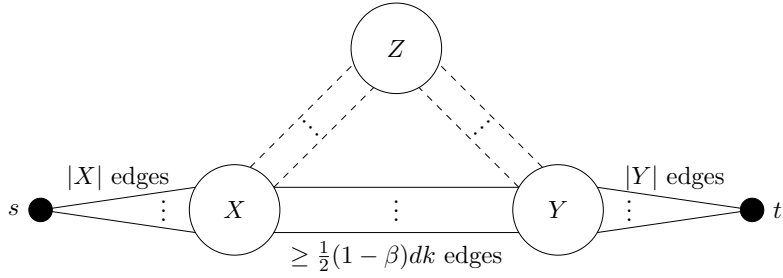 

{\bf {\sc Yes}-case:} 
Suppose $B$ is a {\sc Yes}-instance of $\text{SSE}_\delta(\alpha,\beta)$.
Since $B$ is regular,
there exist subsets $X \subseteq V_X$ and $Y \subseteq V_Y$ such that $|X \cup Y| = \delta|V_X \cup V_Y| = k$ and $\phi_B(X \cup Y) \leq \beta$.
We construct the subgraph $H$ of $G$ as follows.
\begin{itemize}
\item[] {\bf Subgraph $H$:} The subgraph $H$ includes all the edges from $s$ to $X$, all the edges from $X$ to $Y$, and all the edges from $Y$ to $t$.
Since edges from $X$ to $Y$ are of cost zero, the total cost in $H$ is equal to $|X| + |Y| = k$.
\end{itemize}
The following claim will complete the proof of the first item of Proposition~\ref{p:SSE-Reff}.
\begin{lemma}
$\Reff_H(s,t) \leq 2/((1-\beta)dk)$.
\end{lemma}
\begin{proof}
Since $B$ is a $d$-regular bipartite graph, we have 
\[d(|X|+|Y|) = \vol_B(X \cup Y) = |\delta_B(X \cup Y)| + 2 |E_B(X,Y)|,
\]
where $E_B(X,Y)$ denotes the set of edges with one endpoint in $X$ and one endpoint in $Y$. 
Since $\phi_B(X \cup Y) \leq \beta$, we have $|\delta_B(X \cup Y)| \leq \beta \cdot \vol_G(X \cup Y) = d\beta(|X|+|Y|)$. 
Hence, the number of edges between $X$ and $Y$ is
\[|E_B(X,Y)| = \frac{d(|X|+|Y|) - |\delta_B(X \cup Y)|}{2} 
\geq \frac{1}{2}(1-\beta)d(|X|+|Y|) = \frac{1}{2}(1-\beta)dk.\]
In terms of $s$-$t$ effective resistance, $H$ is equivalent to the graph in Figure~\ref{fig:YESReff}, where $Z = (V_X \backslash X) \cup (V_Y \backslash Y)$ is the set of vertices not in $X$ and $Y$. 
Since the edges from $s$ to $X$ and from $Y$ to $t$ have zero resistance and edges between $X$ and $Y$ have resistance one, we have
$\Reff_H(s,t) \leq 2/((1 - \beta) d k)$.
\end{proof}

{\bf {\sc No}-case:}
We will prove the second item of Proposition~\ref{p:SSE-Reff} by arguing that every subgraph of $B$ with total cost at most $k$ has considerably larger $s$-$t$ effective resistance.
Since all the edges between $V_X$ and $V_Y$ have zero cost and adding edges never increases $s$-$t$ effective resistance (by Rayleigh's monotonicity principle), 
we can assume without loss of generality that any solution $H$ to the weighted $s$-$t$ effective resistance network design problem takes all edges between $V_X$ and $V_Y$ and also takes exactly $k$ edges from $E_s \cup E_t$. 
Consider an arbitrary subgraph $H$ with the above properties.
Let $X \subseteq V_X$ be the set of neighbors of $s$ and $Y \subseteq V_Y$ be the set of neighbors of $t$, with $|X|+|Y|=k$.
Let $\phi:=\phi_B(X \cup Y)$.
Note that $\phi \geq \alpha$ as we are in the {\sc No}-case where $\phi_B(X \cup Y) \geq \alpha$ for every $|X \cup Y|=k$.
Using the same calculation as above, we have
\[
|E_B(X,Y)| = \frac{1}{2} (1-\phi_B(X \cup Y))dk = \frac{1}{2} (1-\phi)dk.
\]
The subgraph $H$ is shown in Figure~\ref{fig:NOReff}, where $Z = (V_X \backslash X) \cup (V_Y \backslash Y)$ is the set of vertices not in $X$ and $Y$, and the edges within $Z$ are not shown. 
To lower bound $\Reff_H(s,t)$, we modify $H$ to obtain $H'$ and argue that $\Reff_H(s,t) \geq \Reff_{H'}(s,t)$ and then show a lower bound on $\Reff_{H'}(s,t)$.

\begin{figure}[!ht]
\centering
\resizebox{0.65\textwidth}{!}{
\begin{tikzpicture}
	\def \rad{.7}
    \def \phi{30}
    \def \sx{0}
    \def \sy{0}
    \def \tx{11}
    \def \ty{0}
    \def \Xx{3}
    \def \Xy{0}
    \def \Yx{8}
    \def \Yy{0}
    \def \Zx{5.5}
    \def \Zy{2.5}
	
	\node[draw, radius=\radst, circle, fill, label=left:$s$] (s) at (\sx,\sy) {};
    \node[draw, radius=\radst, circle, fill, label=right:$t$] (t) at (\tx,\ty) {};
    \node (X) at (\Xx,\Xy) {$X$};
    \draw (\Xx,\Xy) circle [radius=\rad];
    \node (Y) at (\Yx,\Yy) {$Y$};
    \draw (\Yx,\Yy) circle [radius=\rad];
    \node (Z) at (\Zx,\Zy) {$Z$};
    \draw (\Zx,\Zy) circle [radius=\rad];
    
    \draw ($(\sx,\sy)$) -- node[above] {$|X|$ edges} ($(\Xx,\Xy) + ( 180-\phi: \rad )$);
    \node (dots) at ($(\sx/2,\sy/2) + (\Xx/2,\Xy/2) + (.4,0.1)$) {$\vdots$};
    \draw ($(\sx,\sy)$) -- ($(\Xx,\Xy) + ( -180+\phi: \rad )$);
    
    \draw ($(\Yx,\Yy) + ( \phi: \rad )$) -- node[above] {$|Y|$ edges} ($(\tx,\ty)$);
    \node (dots) at ($(\Yx/2,\Yy/2) + (\tx/2,\ty/2) + (-.4,0.1)$) {$\vdots$};
    \draw ($(\Yx,\Yy) + ( -\phi: \rad )$) -- ($(\tx,\ty)$);
    
    \draw ($(\Xx,\Xy) + ( \phi-10: \rad )$) --  ($(\Yx,\Yy) + ( 180-\phi+10: \rad )$);
    \node (dots) at ($(\Xx/2,\Xy/2) + (\Yx/2,\Yy/2) + (0, 0.1)$) {$\vdots$};
    \draw ($(\Xx,\Xy) + ( -\phi+10: \rad )$) -- node[below] {$\leq \frac{1}{2} (1-\alpha)  d  k$ edges} ($(\Yx,\Yy) + ( -180+\phi-10: \rad )$);
    
    \draw ($(\Xx,\Xy) + ( \phi-10: \rad )$) -- ($(\Zx,\Zy) + ( -80-\phi: \rad )$);
    \draw ($(\Xx,\Xy) + ( 45+2*\phi: \rad )$) -- ($(\Zx,\Zy) + ( -135-2*\phi: \rad )$);
    \node [rotate=-135] (dots) at ($(\Xx/2,\Xy/2) + (\Zx/2,\Zy/2) - (.1,0)$) {$\vdots$};

    \draw ($(\Yx,\Yy) + ( 190-\phi: \rad )$) -- ($(\Zx,\Zy) + ( -100+\phi: \rad )$);
    \draw ($(\Yx,\Yy) + ( 45+\phi: \rad )$) -- ($(\Zx,\Zy) + ( \phi/2: \rad )$);
    \node [rotate=135] (dots) at ($(\Yx/2,\Yy/2) + (\Zx/2,\Zy/2) + (.1,0)$) {$\vdots$};
\end{tikzpicture}
}
\caption{The subgraph $H'$ is obtained by identifying the subsets $X,Y,Z$ into single vertices.} \label{fig:NOReff}
\end{figure}
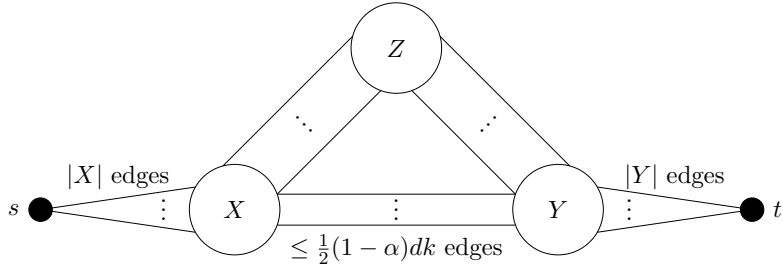

To obtain $H'$ from $H$, we simply identify the three subsets of vertices $X,Y,Z$ to three vertices, which is equivalent to adding a clique of zero resistance edges to each of these three subsets.
By Rayleigh's monotonicity principle, this could only decrease the $s$-$t$ effective resistance and so we have $\Reff_H(s,t) \geq \Reff_{H'}(s,t)$.

In terms of $s$-$t$ effective resistance,
the subgraph $H'$ is equivalent to the graph with two paths between $X$ and $Y$ (with parallel edges): one path $P_1$ of length one with $|E_B(X,Y)|$ parallel edges between $X$ and $Y$, another path $P_2$ of length two with $|E_B(X,Z)|$ parallel edges between $X$ and $Z$ and $|E_B(Z,Y)|$ parallel edges between $Z$ and $Y$.
To lower bound $\Reff_{H'}(s,t)$, we lower bound the resistance of $P_1$ and $P_2$, denoted by $r(P_1)$ and $r(P_2)$.
Note that 
\[r(P_1) = \frac{1}{E_B(X,Y)} = \frac{2}{(1-\phi)dk}.\]
For $r(P_2)$, 
let $x = |\delta_B(X,Z)|$ and $y = |\delta_B(Y,Z)|$, 
then
\[
r(P_2) = \frac{1}{x} + \frac{1}{y} = \frac{1}{x+y} \cdot \frac{(x+y)^2}{xy} = \frac{1}{x+y} \cdot \left(\frac{x}{y} + \frac{y}{x} + 2\right) \geq \frac{4}{x+y} = \frac{4}{\phi dk},
\]
where the inequality holds since $a + 1/a \geq 2$ for any $a > 0$,
and the last equality holds because $x + y = |\delta_B(X \cup Y)| = \phi dk$.
Finally, by Fact~\ref{f:resistance-SP},
\[
\Reff_H(s,t) \geq \Reff_{H'}(s,t) = \frac{1}{1/r(P_1) + 1/r(P_2)}
\geq \frac{1}{\frac{1}{2} (1-\phi) dk + \frac{1}{4} \phi dk} 
= \frac{2}{(1-\phi/2)dk}
\geq \frac{2}{(1-\alpha/2)dk},
\]
where the last inequality is because we are in the {\sc No}-case.
This completes the proof of the second item of Proposition~\ref{p:SSE-Reff}.

\begin{remark}
In this subsection, we show the hardness of the weighted $s$-$t$ effective resistance network design problem, when the edge cost and the edge resistance could be arbitrary.
Using a similar argument as in the proof of Theorem~\ref{t:NPc},
the reduction can be modified to the unit-cost case if we replace the edges from $s$ to $V_X$ and $V_Y$ to $t$ by sufficiently long paths (so that the cost of connecting $s$ to a vertex in $V_X$ is much larger than the cost of connecting a vertex in $V_X$ to a vertex in $V_Y$).
Therefore, the same $(2 - \eps)$-SSE-hardness also holds in the case when every edge has the same cost.
\end{remark}

\subsubsection{From Small Set Expansion to Bipartite Small Set Expansion} \label{ss:bipartite}

We prove Proposition~\ref{p:bipartite} in this subsection.

{\bf Construction:}
Given an instance $\SSE_\delta(1-\eps,\eps)$ on a $d$-regular graph $G=(V,E)$, we construct a $d$-regular bipartite graph $B=(V_X,V_Y;E_B)$ as follows.
For each vertex $v$ in $V$, we create a vertex $v_X \in V_X$ and a vertex $v_Y \in V_Y$, so that $|V_X|=|V_Y|=|V|$.
For each edge $uv \in E$, we add two edges $u_X v_Y$ and $u_Y v_X$ to $E_B$.
It is clear from the construction that $B$ is $d$-regular. 

{\bf Correctness:}
To prove Proposition~\ref{p:bipartite}, we will establish the following two claims.
\begin{enumerate}
\item {\bf {\sc Yes}-case:} 
If there is a set $S \subseteq V$ with $|S| = \delta |V|$ and $\phi_G(S) \leq \eps$ in $G$, then there exist $X \subseteq V_X$ and $Y \subseteq V_Y$ with $|X|+|Y| = \delta(|V_X|+|V_Y|)$ and $\phi_B(X \cup Y) \leq \eps$ in $B$.
\item {\bf {\sc No}-case:}
If every set $S \subseteq V$ with $|S| = \delta |V|$ has $\phi_G(S) \geq 1-\eps$ in $G$, then every sets $X \subseteq V_X$ and $Y \subseteq V_Y$ with $|X|+|Y| = \delta(|V_X|+|V_Y|)$ has $\phi_B(X \cup Y) \geq 1-8\eps$ in $B$.
\end{enumerate}

{\bf {\sc Yes}-case:}
Let $S \subseteq V$ be a subset with $|S|=\delta|V|$ and $\phi_G(S) \leq \eps$ in $G$.
Let $S_X := \{v_X \mid v \in S\}$ and $S_Y := \{v_Y \mid v \in S\}$, with $|S|=|S_X|=|S_Y|$.
By construction, an edge $uv \in \delta_G(S)$ if and only if both $u_X v_Y$ and $v_X u_Y$ are in $\delta_B(S_X \cup S_Y)$, and thus $|\delta_B(S_X \cup S_Y)| = 2|\delta_G(S)|$.
Since $|S_X \cup S_Y| = |S_X|+|S_Y| = 2|S|$ and $B$ is $d$-regular, we have
\[
\phi_B(S_X \cup S_Y) = \frac{|\delta_B(S_X \cup S_Y)|}{\vol_B(S_X \cup S_Y)}
= \frac{|\delta_B(S_X \cup S_Y)|}{d(|S_X| + |S_Y|)}
= \frac{2|\delta_G(S)|}{2d|S|} 
= \phi_G(S)
\leq \eps.
\]

{\bf {\sc No}-case:}
Consider arbitrary subsets $X \subseteq V_X$ and $Y \subseteq V_Y$ with $|X|+|Y| = \delta(|V_X|+|V_Y|) = 2\delta|V|$.
To lower bound $\phi_B(X \cup Y)$, we will upper bound $|E_B(X,Y)|$.
We partition $X$ into groups $X_1, \ldots, X_a$ where every group except the last group is of size $\delta|V|/2$ and the last group is of size at most $\delta|V|/2$.
We partition $Y$ into groups $Y_1, \ldots, Y_b$ in a similar way.
The following claim uses the small-set expansion property in $G$ to show that there is no small dense subset in $B$.

\begin{lemma} \label{f:expansion}
Suppose $G$ is a {\sc No}-instance of $\SSE_\delta(1-\eps, \eps)$.
Then, for any $1 \leq i \leq a$ and $1 \leq j \leq b$, 
\[|E_B(X_i, Y_j)| \leq \eps \delta d |V|.\]
\end{lemma}
\begin{proof}
We first argue that there is no small dense subset in $G$, and then we will use it to bound $|E_B(X_i,Y_j)|$.
Suppose $S \subseteq V$ with $|S| = \delta|V|$.
As $G$ is a {\sc No}-instance, we know that $\phi_G(S) \geq 1-\eps$ and thus 
$|\delta_G(S)| \geq (1-\eps)\vol_G(S) = (1-\eps)d|S|$.
Since $d|S| = \vol_G(S) = |\delta_G(S)| + 2|E_G(S,S)|$,
it follows that $|E_G(S,S)| \leq \eps d|S|/2 = \eps \delta d |V|/2$.
Note that this also implies trivially that $|E_G(Z,Z)| \leq \eps \delta d |V|/2$ for any $Z$ with $|Z| \leq \delta|V|$.

Given $X_i$ and $Y_j$, let $Z := \{v \in G \mid v_X \in X_i {\rm~or~} v_Y \in Y_j\}$.
In words, $Z$ is the set of vertices in $G$ which have at least one copy in $X_i \cup Y_j$ in $B$.
Since each $X_i$ and $Y_j$ is of size at most $\delta |V|/2$,
it follows that $|Z| \leq \delta |V|$.
Also, note that $|E_B(X_i,Y_j)| \leq 2|E_G(Z,Z)|$,
as each edge in $E_B(X_i,Y_j)$ corresponds to one edge in $E_G(Z,Z)$ while each edge in $E_G(Z,Z)$ is corresponded to at most two edges in $E_B(X_i,Y_j)$.
Therefore, we can apply the bound in the previous paragraph to conclude that $|E(X_i,Y_j)| \leq 2|E_G(Z,Z)| \leq \eps \delta d |V|$.
\end{proof}

We now use the lemma to bound $|E_B(X,Y)|$.
Since $|X|+|Y|=2\delta|V|$, it follows that $a \leq 4$ and $b \leq 4$,
and therefore 
\[|E_B(X,Y)| \leq \sum_{i=1}^a \sum_{j=1}^b |E_B(X_i,Y_j)| \leq ab\eps \delta d|V| \leq 16\eps \delta d |V|.\]
As $B$ is bipartite,
\[|\delta_B(X \cup Y)| = \vol_B(X \cup Y) - 2|E_B(X,Y)|
\geq 2\delta d|V| - 32 \eps \delta d |V| = 2(1-16\eps)\delta d |V|.
\]
Therefore, we have
\[\phi_B(X \cup Y) = \frac{|\delta_B(X \cup Y)|}{\vol_B(X \cup Y)}
\geq \frac{2(1-16\eps)\delta d |V|}{2\delta d |V|} = 1-16\eps.
\]
This completes the proof of Proposition~\ref{p:bipartite}.
We remark that a more careful argument gives $|E_B(X,Y)| \leq 6\eps \delta d|V|$ and thus $\phi_B(X \cup Y) \geq 1-6\eps$, but this constant does not matter for the proof of Theorem~\ref{t:APX}.


\section{Concluding Remarks}

We have formulated a new and natural network design problem and presented some hardness and algorithmic results.
It opens up a number of interesting problems to be studied.
\begin{enumerate}
\item 
For the $s$-$t$ effective resistance network design problem, we conjecture that the integrality gap of the convex program is exactly two.
As mentioned in Remark~\ref{r:five}, the analysis of the $8$-approximation is not tight, and we can show that the same algorithm achieves an approximation ratio strictly smaller than $5$.
It would be good to close the gap completely.
\item 
The weighted case of arbitrary costs and arbitrary resistances is wide open.  
It will be interesting if there are stronger convex programming relaxations for the problem (perhaps adding some knapsack constraint as suggested by the dynamic programming algorithms for series-parallel graphs).
\item
As in survivable network design, one could study the general problem when there are multiple source-sink pairs and each pair has a different effective resistance requirement.
It will be very interesting if it is still possible to achieve a constant factor approximation in this general setting.
\item
An interesting intermediate problem is to find a minimum cost network so that the maximum effective resistance over pairs (the resistance diameter) is minimized. 
This is an analog of the global connectivity problem in traditional network design.
\end{enumerate}
A more open-ended direction is to unify and extend the techniques for network design problems with spectral requirements.

\bibliographystyle{plain}

\end{document}